\documentclass{llncs}
\usepackage{mathptmx}
\usepackage{latexsym}
\usepackage{times}
\usepackage{amsmath}
\usepackage{amssymb}

\usepackage{color,graphicx,subfig}
\usepackage{epic}
\usepackage{fancybox,graphpap}
\usepackage{pb-diagram,pb-xy}
\usepackage[all]{xy}
\usepackage{bbm}
\usepackage{array}
\usepackage{algorithmic}
\usepackage{hyperref}

\begin{document}

\newcommand{\END}{\qed}
\newcommand{\LT}{\mathcal{L}}
\newcommand{\eqa}{\thickapprox}
\newcommand{\eqb}{\equiv}
\newcommand{\PS}[1]{\wp(#1)}
\newcommand{\PSB}[1]{\wp^{\setminus \set{\emptyset}}(#1)}

\newcommand{\h}[1]{\overline{#1}}
\newcommand{\set}[1]{\{#1\}}
\newcommand{\bset}[1]{\bigl\{#1\bigr\}}
\newcommand{\Bset}[1]{\Bigl\{#1\Bigr\}}
\newcommand{\Al}{\biguplus}
\newcommand{\RC}{\div_R}
\newcommand{\LC}{\div_L}
\newcommand{\wei}{\textit{weight}}
\newcommand{\LG}[1]{\langle #1 \rangle}
\newcommand{\E}{\mathbb{S}}
\newcommand{\EC}{\widehat{\mathcal{S}}}
\newcommand{\PO}{\prec}
\newcommand{\st}{\mathsf{st}}
\newcommand{\com}{<\!\!>}
\newcommand{\df}{\triangleq}
\newcommand{\iffdf}{\stackrel{\textit{\scriptsize{df}}}{\iff}\ }
\newcommand{\calf}[1]{\mathcal{#1}}
\newcommand{\sq}{\sqsubset}
\newcommand{\ccl}{\;\bowtie}
\newcommand{\todo}[1]{ \textcolor{red}{TODO: #1}}
\newcommand{\tcomment}[1]{\text{\hspace*{2mm}$\langle$~\parbox[t]{\textwidth}{ #1 $\rangle$}}}
\newcommand{\ttcomment}[1]{\text{$\langle$~#1~$\rangle$}}
\newcommand{\sym}[1]{{#1}^{\mathsf{sym}\;}}
\newcommand{\imm}[1]{{#1}^{\mathsf{cov}\;}}

\newcommand{\TR}{\mathsf{TR}}
\newcommand{\LCT}{\mathsf{LCT}}
\newcommand{\DCD}{\mathsf{CDG}}
\newcommand{\andspace}{\hspace*{3mm}\text{ and }\hspace*{3mm}}
\newcommand{\myspace}[1]{\mbox{\hspace{#1cm}}}
\newcommand{\map}{\mathsf{map}}
\newcommand{\bt}{\mathsf{bt}}
\newcommand{\cl}{\mathsf{ct2lct}}
\newcommand{\lc}{\mathsf{lct2ct}}
\newcommand{\cd}{\mathsf{ct2dep}}
\newcommand{\dl}{\mathsf{dep2lct}}
\newcommand{\ld}{\mathsf{lct2dep}}

\newcommand{\lhdf}{\lhd^\frown}
\newcommand{\flhd}{\frown_\lhd}
\newcommand{\lex}[1]{\,{#1}^{\textit{lex}}\,}
\newcommand{\optord}{\,<^{\textit{opt}}\,}
\newcommand{\stor}[1]{\,{#1}^{\textit{st}}\,}
\newcommand{\reco}[1]{ {#1}^{\,{\textsf{C}}} }
\newcommand{\si}[1]{(#1)^\Cap}
\newcommand{\CT}[1]{\mathfrak{C}(#1)}
\newcommand{\GCT}[1]{\mathfrak{gC}(#1)}
\newcommand{\eqrel}[1]{\equiv_{#1}}
\newcommand{\quotient}[2]{{#1}/\!{#2}}
\newcommand{\FEC}{\mathbb{C}\mathit{at}_{\mathsf{FE}}}
\newcommand{\cau}{\longrightarrow}
\newcommand{\wcau}{\dashrightarrow}
\newcommand{\bcau}{\rightsquigarrow}

\newcommand{\seq}[1]{\left[   #1 \right] }

\newcommand{\It}[1]{\mathit{#1}}
\newcommand{\defref}[1]{Definition~\ref{def:#1}}
\newcommand{\theoref}[1]{Theorem~\ref{theo:#1}}
\newcommand{\propref}[1]{Proposition~\ref{prop:#1}}
\newcommand{\corref}[1]{Corollary~\ref{cor:#1}}
\newcommand{\lemref}[1]{Lemma~\ref{lem:#1}}
\newcommand{\exref}[1]{Example~\ref{ex:#1}}
\newcommand{\reref}[1]{Remark~\ref{re:#1}}
\newcommand{\eref}[1]{\eqref{eq:#1}}
\newcommand{\figref}[1]{Figure~\ref{fig:#1}}

\newcommand{\secref}[1]{Section~\ref{sec:#1}}
\newcommand{\eqnref}[1]{Eq.~(\ref{eq:#1})}
\newcommand{\TCT}{\stackrel{\mathsf{t}\leftrightsquigarrow \mathsf{c}}{\equiv}}

\newcommand{\EOD}{\hfill {\scriptsize $\blacksquare$}}


\numberwithin{equation}{section}

\title{A Characterization of Combined Traces Using  Labeled Stratified Order Structures}

\author{Dai Tri Man L\^e}
\institute{Department of Computer Science, University of Toronto\\
10 King's College Road, Toronto, ON, M5S 3G4 Canada\\
\email{ledt@cs.toronto.edu}}

\maketitle
\pagestyle{headings} 
\begin{abstract}
This paper defines a class of labeled stratified order structures that characterizes exactly the notion of \emph{com}bined \emph{traces} (i.e., \emph{comtraces}) proposed by Janicki and Koutny in 1995. Our main technical contributions are the representation theorems showing that comtrace quotient monoid,  \emph{combined dependency graph} (Kleijn and Koutny 2008) and our labeled stratified order structure characterization are three different and yet equivalent ways to represent comtraces.
\end{abstract}

\section{Introduction}
Partial orders are a principle tool for modelling ``true concurrency'' semantics of concurrent systems (cf. \cite{Pra}). They are utilized to develop powerful partial-order based automatic verification techniques, e.g.,  \emph{partial order reduction} for model checking concurrent software (see, e.g., \cite[Chapter 10]{CGP} and \cite{EH}). Partial orders are also equipped with \emph{traces}, their powerful formal language counterpart,  proposed by Mazurkiewicz \cite{Ma1}. In  \emph{The Book of Traces} \cite{DR}, trace theory has been used to tackle problems from diverse areas including formal language theory, combinatorics, graph theory, algebra, logic, and \emph{concurrency theory}.

However, while partial orders and traces can sufficiently model the ``earlier than" relationship, Janicki and Koutny argued that it is problematic to use a single partial order to specify both  the ``earlier than" and  the ``not later than" relationships \cite{J4}. This motivates them to develop the theory of \emph{relational structures}, where a pair of relations is used to capture concurrent behaviors.  The most well-known among the classes of relational structures  proposed by Janicki and Koutny is the class of \emph{stratified order structures} (\emph{so-structures}) \cite{GP,JK0,JK95,JK97,J0}. A so-structure is a triple $(X,\prec,\sqsubset)$, where $\prec$ and $\sqsubset$ are binary relations on $X$. They were invented to model both the ``earlier than" (the relation $\prec$) and ``not later than" (the relation $\sqsubset$) relationships, under the assumption that system runs are described by \emph{stratified partial orders}, i.e., step sequences. They have been successfully applied to model inhibitor and priority systems, asynchronous races, synthesis problems, etc. (see for example \cite{JK95,JK99,JLM06,JLM08,KK,KK08} and others). 

The \emph{com}bined \emph{trace} (\emph{comtrace}) notion, introduced by Janicki and Koutny \cite{JK95}, generalizes the trace notion by utilizing step sequences instead of words. First the set of all possible steps that generates step sequences are identified by a relation $sim$, which is called {\em simultaneity}. Second a congruence relation is determined by a relation $ser$, which is called {\em serializability} and in general \emph{not} symmetric. Then a comtrace is defined as a finite set of congruent step sequences. Comtraces were introduced as a formal language representation of so-structures to provide an operational semantics for  Petri nets with inhibitor arcs. Unfortunately, comtraces have been less often known and applied than so-structures, even though in many cases they appear to be more natural. We believe one  reason is that the comtrace notion was too succinctly discussed in \cite{JK95} without a full treatment dedicated to comtrace theory. Motivated by this, Janicki and the author have devoted our recent effort on the study of comtraces \cite{JL08,Le,JL09}, yet there are too many different aspects to explore and the truth is we can barely scratch the surface. In particular, the huge amount of results from  trace theory (e.g., from \cite{DR,DM}) desperately needs to be generalized to comtraces. These  tasks are often non-trivial since we are required to develop intuition and novel techniques to deal with the complex interactions of the ``earlier than'' and ``not later than'' relations.

This paper gives a novel characterization of comtraces using labeled so-structures. Such definition is interesting for the following reasons. 

First, it defines exactly the class of labeled so-structures that can be represented by comtraces. It is worth noting that this point is particularly important. Even though it was shown in \cite{JK95} that every comtrace can be represented by a labeled so-structure, the converse could not be shown because a class of labeled so-structures that defines precisely the class of  comtraces was not known. The closest to our characterization is  the \emph{combined dependency graph} (\emph{cd-graph}) notion (analogous to \emph{dependence graph} representation of traces) introduced recently by Kleijn and Koutny \cite{KK08}, but again a theorem showing that combined dependency graphs can be represented by comtraces was not given. Our approach is quite different and based on some new ideas discussed in Section 4 of this paper. 

Second, even though the step sequence definition of comtraces is more suitable when dealing with formal language aspects of comtraces, the labeled so-structure representation is more suitable for a variety of powerful order-theoretic results and techniques available to us (cf. \cite{Fis,DP02,J0}).

Finally, the labeled so-structure definition of comtrace can be easily extended to \emph{infinite comtraces}, which describe nonterminating concurrent processes. The labeled poset representation of infinite traces is already successfully applied in both theory and practice, e.g.,  \cite{TW02,FM06,FM07,GGH09}. Although such definition is equivalent to  the one using quotient monoid over infinite words \cite{Gas90,Die91}, we believe that infinite labeled posets are sometimes simpler. Indeed the celebrated work by Thiagarajan and Walukiewicz (cf. \cite{TW02}) on linear temporal logic for traces utilizes the labeled poset characterization of infinite traces, where \emph{configurations} of a trace are conveniently defined as \emph{finite downward closed subsets} of the labeled poset representation. We will not analyze infinite comtraces or logics for comtraces in this paper, but these are fruitful directions to explore using the results from this paper.

The paper is organized as follows. In Section 2, we recall some preliminary definitions and notations. In Section 3, we give a concise exposition of the theory of so-structures and comtraces by Janicki and  Koutny \cite{JK95,JK97}. In Section 4, we give our definition of comtraces using labeled so-structure and some remarks on how we arrived  at such definition. In Section 5, we prove a representation theorem showing that our comtrace definition and the one by Janicki and Koutny are indeed equivalent; then using this theorem, we prove another representation theorem showing that our definition is also equivalent to the cd-graph definition from \cite{KK08}. In Section 6, we define  \emph{composition} operators for our comtrace representation and for cd-graphs. Finally, in Section 7, some final remarks and future works are presented.

\section{Notations \label{sec:background}}

\subsection{Relations, Orders and Equivalences}

The \emph{powerset} of a set $X$ will be denoted by $\PS{X}$, i.e.
$\PS{X}\df\set{Y \mid Y\subseteq X}$. The set of all {\em non-empty} subsets of $X$
will be denoted by $\PSB{X}$. In other words,
$\PSB{X}\df\PS{X}\setminus\set{\emptyset}.$ 

We let $\mathit{id}_X$ denote the \emph{identity relation} on a set $X$. If $R$ and $S$ are binary relations on a set $X$ (i.e., $R,S\subseteq X\times X$), then their \emph{composition} $R\circ S$ is defined as $R\circ S\df\set{(x,y)\in X\times X\mid \exists z\in X.\;(x,z)\in R\wedge (z,y)\in S}$. We also define
\begin{align*}
R^{0}&\df id_X & R^{i}&\df R^{i-1}\circ R\hspace*{3mm} (\text{for }i\ge 1) &R^{+}&\df \bigcup_{i\ge 1}R^i	&R^{*}&\df \bigcup_{i\ge 0}R^i
\end{align*}
The relations $R^{+}$ and $R^{*}$ are called the \emph{(irreflexive) transitive closure} and \emph{reflexive transitive closure} of $R$ respectively.

A binary relation $R\subseteq X\times X$ is an \emph{equivalence relation} relation on $X$ if and only if (iff) $R$ is {\em reflexive},  {\em symmetric} and {\em transitive}. If $R$ is an equivalence relation, then for every $x\in X$, the set $[x]_R \df \set{y\;|\;y\;R\;x \wedge y\in X}$ is the equivalence class of $x$ with respect to $R$. We also define $X/R\df\set{[x]_R\mid x\in X}$, i.e., the set of all equivalence classes of $X$ under $R$. We drop the subscript and write $[x]$ when $R$ is clear from the context. 

A binary relation $\prec \;\subseteq X \times X$ is a
{\em partial order} iff $R$ is {\em irreflexive} and {\em transitive}.
The pair $(X,\prec)$ in this case is called a \emph{partially ordered set} (\emph{poset}). The pair $(X,\prec)$ is called a \emph{finite poset} if $X$ is finite. For convenience, we define:
\begin{align*}
\simeq_\prec  &\df  \bset{(a,b)\in X\times X\mid a\not\prec b \;\wedge\; b\not\prec a}& \text{\emph{(incomparable)}}\\
\frown_\prec  &\df \bset{(a,b)\in X\times X\mid  a \simeq_\prec b \;\wedge\; a\neq b} &
\text{\emph{(distinctly incomparable)}}\\
\prec^\frown  &\df \bset{(a,b)\in X\times X\mid a \prec b \;\vee\; a\frown_\prec b} & \text{\emph{(not greater)}}
\end{align*}

A poset $(X,\prec)$ is {\em total} iff $\frown_\prec$ is empty; and {\em stratified} iff $\simeq_\prec$ is an equivalence relation. Evidently every total order is stratified.

\subsection{Step Sequences\label{sec:steps}}
For every finite set $X$, a set $\E \subseteq \PSB{X}$ can be seen as an alphabet. 
The elements of $\E$ are called {\em steps} and
the elements of $\E^*$ are called {\em step sequences}. For example, if the set of possible steps is 
$\E = \bset{ \{a,b,c\}, \{a,b\}, \{a\}, \{c\} }$, then
$\{a,b\}\{c\}\{a,b,c\} \in \E^*$ is a step sequence.
The triple $(\E^*,\ast,\epsilon)$, where $\_\,\ast\,\_$ denotes the step sequence concatenation operator (usually omitted) and $\epsilon$ denotes the empty step sequence, is a monoid.

Let $t=A_1\ldots A_k$ be a step sequence. We define  $|t|_a$, the number of occurrences of an event $a$ in $w$, as $|t|_a \df \bigl\lvert\bset{A_i\mid 1\le i \le k \wedge a\in A_i}\bigr\rvert$, where $|X|$ denotes the cardinality of the set $X$.  
Then we can construct its unique \emph{enumerated step sequence} $\h{t}$ as\\
\mbox{\hspace{2.5cm}}$\h{t}\df \h{A_1}\ldots \h{A_k}\text{, where } \h{A_i} \df \Bset{e^{(|A_1\ldots A_{i-1}|_e+1)}\bigl\lvert e\in A_i\bigr.}.$\\
We will call such $\alpha=e^{(j)}\in \h{A_i} $ an \emph{event occurrence} of $e$. For instance, if we let  $t= \set{a,b}\set{b,c}\set{c,a}\set{a}$,
then  $\overline{t}=\bset{a^{(1)},b^{(1)}}\bset{b^{(2)},c^{(1)}}\bset{a^{(2)},c^{(2)}}\bset{a^{(3)}}.$ 

We let $\Sigma_t=\bigcup_{i=1}^k \h{A_i}$ denote the set of all event occurrences in all steps of $t$. For example, when
$t= \set{a,b}\set{b,c}\set{c,a}\set{a}$, $\Sigma_t=\bset{ a^{(1)},a^{(2)},a^{(3)}, b^{(1)}, b^{(2)},
c^{(1)},c^{(2)} }.$ We also define $\ell:\Sigma_t\rightarrow E$ to be the function that returns the \emph{label} of $\alpha$ for each $\alpha \in \Sigma_t$. For example, if $\alpha=e^{(j)}$, then $\ell(\alpha)=\ell(e^{(j)})=e$. Hence, from an enumerated step sequence $\h{t} = \h{A_1}\ldots \h{A_k}$, we can uniquely reconstruct its step sequence $t =  \ell(\,\h{A_1}\,)\ldots \ell(\,\h{A_k}\,).$

For each $\alpha \in \Sigma_u$, we let $pos_t(\alpha)$ denote the consecutive number of a step where $\alpha$ belongs, i.e., if $\alpha\in \h{A_i}$ then $pos_t(\alpha)=i$. For our example, $pos_t(a^{(2)})=3$,  $pos_t(b^{(2)})=pos_t(c^{(1)})=2$, etc.

It is important to observe that step sequences and stratified orders are interchangeable concepts. Given a step sequence $u$, define the binary relation $\lhd_u$  on $\Sigma_u$ as\\
\mbox{\hspace{3.6cm}}$\alpha \lhd_u \beta \iffdf pos_u(\alpha)<pos_u(\beta).$

Intuitively, $\alpha \lhd_u \beta$ simply means $\alpha$ occurs before $\beta$ on the step sequence $u$. Thus, $\alpha \lhd_u^\frown \beta$ iff $(\alpha\not=\beta\wedge pos_u(\alpha)\le pos_u(\beta))$; and $\alpha\simeq_u \beta$ iff  $pos_u(\alpha)=pos_u(\beta)$. Obviously, the relation $\lhd_u$ is a stratified order and we will call it the stratified order {\em generated by the step sequence} $u$.

Conversely, let $\lhd$ be a stratified order on a set $\Sigma$. The set $\Sigma$ can be represented as a
sequence of equivalence classes $\Omega_\lhd=B_1\ldots B_k$ ($k\ge 0$) such that \[\lhd = \bigcup_{i<j}B_i\times B_j\text{\mbox{\hspace{4mm}} and \mbox{\hspace{4mm}} }\simeq_\lhd \;= \bigcup_{i}B_i\times B_i.\]
The sequence $\Omega_\lhd$ is called the step sequence \emph{representing} $\lhd$. A detailed discussion on this connection between stratified orders and step sequences can be found in \cite{JL09}.

\section{Stratified Order Structures and Combined Traces}
In this section, we review the Janicki -- Koutny theory of stratified order structures and comtraces from \cite{JK95,JK97}. The reader is also referred to \cite{KK08} for an excellent introductory survey on the subject with many motivating examples.

\subsection{Stratified Order Structures}
A \emph{relational structure} is a triple $T=(X,R_1,R_2)$, where $X$ is a set and $R_1$, $R_2$ are binary relations on $X$. A relational structure $T'=(X',R'_1,R'_2)$ is an \emph{extension} of $T$, denoted as $T\subseteq T'$, iff $X=X'$, $R_1\subseteq R_1'$ and $R_2\subseteq R_2'$.

\begin{definition}[stratified order structure \cite{JK97}]
A \emph{stratified order structure} (\emph{so-structure}) is a relational structure $S=(X,\prec,\sqsubset),$ such that for all $\alpha,\beta,\gamma \in X$, the following hold:
\begin{align*}
\text{\textsf{S1:}\hspace{3mm}}& \alpha\not\sqsubset \alpha & \text{\textsf{S3:}\hspace{3mm}}& \alpha\sqsubset \beta \sqsubset \gamma \;\wedge\; \alpha \not= \gamma \implies \alpha\sqsubset \gamma\\
\text{\textsf{S2:}\hspace{3mm}}& \alpha \prec \beta \implies \alpha \sqsubset \beta \hspace{3mm}& \text{\textsf{S4:}\hspace{3mm}}&
\alpha\sqsubset \beta \prec \gamma \;\vee\; \alpha\prec \beta \sqsubset \gamma \implies \alpha\prec \gamma
\end{align*}
When $X$ is finite, $S$ is called a \emph{finite so-structure}. \EOD
\label{def:sos}
\end{definition}

The axioms \textsf{S1}--\textsf{S4} imply that $\PO$ is a partial order and $\alpha\prec \beta\Rightarrow \beta\not\sq \alpha.$  The axioms \textsf{S1} and \textsf{S3} imply $\sq$ is a \emph{strict preorder}. The relation $\prec$ is called \textit{causality} and represents the ``earlier than" relationship while the relation $\sqsubset$ is called \textit{weak causality} and represents the ``not later than" relationship. The axioms \textsf{S1}--\textsf{S4} model the mutual relationship between ``earlier than" and ``not later than" relations, provided that {\em the system runs are stratified orders}.  Historically, the name ``stratified order structure'' came from the fact that stratified orders can be seen as a special kind of so-structures.

\begin{proposition}[\cite{J4}]
For every stratified poset $(X,\lhd)$, the triple $S_\lhd=(X,\lhd,\lhd^\frown)$ is a so-structure.\END
\label{prop:soss}
\end{proposition}

We next recall the notion of \emph{stratified order extension}. This concept is extremely important  since the relationship between stratified orders and so-structures is exactly analogous to the one between total orders and partial orders.

\begin{definition}[stratified extension \cite{JK97}]
Let $S=(X,\prec,\sqsubset)$ be a so-structure. A {\em stratified} order $\lhd$ on $X$ is a {\em stratified extension} of $S$ if and only if $(X,\prec,\sqsubset)\subseteq (X,\lhd,\lhd^{\frown})$. 

The set of all stratified extensions of $S$  is denoted as  $ext(S)$. \EOD
\label{def:extsos}
\end{definition}

Szpilrajn's Theorem \cite{Szp} states that every poset can be reconstructed by taking the intersection of all of its total order extensions. Janicki and Koutny showed that a similar result holds for so-structures and stratified extensions:

\begin{theorem}[{\cite{JK97}}]
Let $S=(X,\PO,\sq)$ be a so-structure. Then\\
\mbox{\hspace{3.5cm}} $ S=\left(X,\bigcap_{\lhd\;\in\; ext(S)}\lhd,\bigcap_{\lhd\;\in\; ext(S)}\lhd^\frown\right).$
\qed
\label{theo:SzpStrat}
\end{theorem}

Using this theorem, we can show the following properties relating so-structures and their stratified extensions.
\begin{corollary}
For every so-structure $S=(X,\PO,\sq)$,
\begin{enumerate}
 \item $\bigl(\exists \lhd\in ext(S),\ \alpha\lhd\beta\bigr)\wedge\bigl(\exists \lhd\in ext(S),\ \beta\lhd\alpha\bigr)\implies \bigl(\exists \lhd\in ext(S),\ \beta\frown_\lhd\alpha\bigr).$
 \item $\bigl(\forall \lhd\in ext(S),\ \alpha\lhd\beta \vee \beta\lhd\alpha\bigr)\iff \alpha \PO \beta \vee \beta \PO \alpha.$
\end{enumerate}
\label{cor:SzpStrat}
\end{corollary}
\begin{proof}\textbf{1. } See \cite[Theorem 3.6]{JK97}. \textbf{2. } Follows from \textbf{1.} and \theoref{SzpStrat}. \qed

%
\end{proof}

\subsection{Combined Traces}

{\em Comtraces} were introduced in \cite{JK95} as a generalization of traces to represent so-structures. The \emph{comtrace congruence} is defined via two relations {\em simultaneity} and {\em serializability}.

\begin{definition}[comtrace alphabet \cite{JK95}] Let $E$ be a finite set (of events) and let $ser \subseteq sim \subset E\times E$ be two relations called \emph{serializability} and \emph{simultaneity} respectively and the relation $sim$ is irreflexive and symmetric. The triple $\theta = (E,sim,ser)$ is called a \emph{comtrace alphabet}. \EOD
\label{def:comalpha}
\end{definition}

Intuitively, if $(a,b)\in sim$ then $a$ and $b$ can occur simultaneously
(or be a part of a {\em synchronous} occurrence in the sense of \cite{JLM06}),
while $(a,b)\in ser$ means that $a$ and $b$ may occur simultaneously
or $a$ may occur before $b$. We define $\E_\theta$, the set of all possible {\em steps},
 to be the set of all cliques of
the graph $(E,sim)$, i.e.,\\
$\mbox{\hspace{2cm}}\E_{\theta} \df \bset{ A \mid A\neq\emptyset \;\wedge\; \forall a,b\in A,\; \bigl(a=b \vee (a,b)\in sim\bigr)}.$

\begin{definition}[comtrace congruence \cite{JK95}]
For a comtrace alphabet $\theta=(E,sim,ser)$, we define $\eqa_{\theta}\;\subseteq\; \E_{\theta}^*\times\E_{\theta}^*$ to be the relation comprising all pairs $(t,u)$ of step sequences such that \smallskip\\
\myspace{4.3}$t=wAz\andspace u=wBCz,$\smallskip\\
where $w,z\in\E_{\theta}^*$ and $A$, $B$, $C$ are steps satisfying $B\cup C \;= \;A$ and $B\times C\;\subseteq\; ser$. 

We define \emph{comtrace congruence} $\eqb_{\theta}\df \left(\eqa_{\theta}\cup\eqa^{-1}_{\theta}\right)^*$. We define the comtrace concatenation operator $\_\circledast\_$ as $[r]\circledast[t] \df [r\ast t]$. The quotient monoid $(\E^*/\!\!\equiv_{\theta},\circledast,[\epsilon])$ is called the monoid of {\em comtraces} over $\theta$. \EOD
\label{def:commonoid}
\end{definition}

Note that since $ser$ is irreflexive, $B\times C\subseteq ser$ implies that $B\cap C=\emptyset$. We will omit the subscript $\theta$ from the comtrace congruence $\eqa_{\theta}$, and write $\equiv$ and $\eqa$ when it causes no ambiguity. To shorten our notations, we often write $[s]_{\theta}$ or $[s]$ instead of $[s]_{\eqb_{\theta}}$ to denote the comtrace generated by the step sequence $s$ over $\theta$.

\begin{example}
Let $E=\set{a,b,c}$ where $a$, $b$ and $c$ are three atomic operations, where\smallskip\\
\myspace{2}$a:\;\; y\leftarrow x+y \myspace{1,5} b:\;\; x\leftarrow y+2\myspace{1.5}  c:\;\; y\leftarrow y+1$\smallskip\\
Assume simultaneous reading is allowed. Then only $b$ and $c$ can be performed simultaneously, and the simultaneous execution of $b$ and $c$ gives the same outcome as executing $b$ followed by $c$.
 We can then define the comtrace alphabet $\theta=(E,sim,ser)$, where
$sim = \bset{\set{b,c}}$ and $ser=\set{(b,c)}$. This yields $\E_{\theta}=\bset{\{a\},\{b\},\{c\},\set{b,c}}$.
Thus, $\textbf{t}=[\{a\}\{b,c\}] =\bset{ \{a\}\{b,c\},\{a\}\{b\}\{c\}}$ is a comtrace.
But $\{a\}\{c\}\{b\} \notin \textbf{t}$. \EOD
\label{ex:comtrace1}
\end{example}

Even though traces are quotient monoids over sequences and comtraces are
quotient monoids over step sequences, traces can be regarded as special kinds
of comtraces when the relation $ser=sim$. For a more detailed discussion on this connection between traces and comtraces, the reader is referred to \cite{JL09}.

\begin{definition}[\cite{JK95}]
Let $u\in \E_{\theta}^*$. We define the relations $\prec_u,\sqsubset_u \subseteq \Sigma_{u}\times \Sigma_{u}$
as:
\begin{enumerate}
\item
$\alpha \prec_u \beta \iffdf \alpha \lhd_u \beta \wedge (\ell(\alpha),\ell(\beta))\notin ser,$

\item
$ \alpha \sqsubset_u \beta \iffdf \alpha \lhd_u^\frown \beta \wedge (\ell(\beta),\ell(\alpha))\notin ser$. \EOD
\end{enumerate}

\label{def:s2inv}
\end{definition}

It is worth noting that the structure $( \Sigma_{u}, \prec_u ,\sqsubset_u, \ell )$ is exactly the \emph{cd-graph} (cf. \defref{comdag}) that represents the comtrace $[u]$. This gives us some intuition on how Koutny and Kleijn constructed the cd-graph definition in \cite{KK08}. We  also observe that $( \Sigma_{u}, \prec_u ,\sqsubset_u )$ is usually \emph{not} a so-structure since $\PO_u$ and $\sq_u$ describe only basic ``local'' causality and weak causality invariants of the event occurrences of $u$ by considering pairwise serializable relationships of event occurrences. Hence, $\PO_u$ and $\sq_u$ might not capture ``global'' invariants that can be inferred from \textsf{S2}--\textsf{S4} of \defref{sos}. To ensure  all invariants are included, we need the following $\lozenge$-closure operator. 

\begin{definition}[\cite{JK95}]
For every relational structure $S=(X,R_1,R_2)$ we define $S^\lozenge$ as\smallskip\\
\mbox{\hspace{2cm}}$S^\lozenge \df \bigl(X,(R_1\cup R_2)^* \circ R_1 \circ (R_1\cup R_2)^*,(R_1\cup R_2)^* \setminus id_X\bigr).$ \EOD
\label{def:SO-CL}
\end{definition}

Intuitively $\lozenge$-closure is a generalization of transitive closure for relations to relational structures. The motivation is that for appropriate  relations $R_1$ and $R_2$ (see assertion (3) of \propref{so-cl}), the relational structure $(X,R_1,R_2)^\lozenge$ is a so-structure. The $\lozenge$-closure operator satisfies the following properties:

\begin{proposition}[\cite{JK95}]
Let $S=(X,R_1,R_2)$ be a relational structure.
\begin{enumerate}
\item If $R_2$ is irreflexive then $S\subseteq S^\lozenge$.
\item $(S^\lozenge)^\lozenge = S^\lozenge$.
\item $S^\lozenge$ is a so-structure if and only if 
$(R_1\cup R_2)^* \circ R_1 \circ (R_1\cup R_2)^*$
is irreflexive.
\item If $S$ is a so-structure then $S=S^\lozenge$.
\item If $S$ be a so-structure and $S_0 \subseteq S$, then $S_0^{\lozenge}\subseteq S$ and $S_0^{\lozenge}$ is a so-structure. \hspace{-5mm}
\END
\end{enumerate}
\label{prop:so-cl}
\end{proposition}

\begin{definition}
Given a step sequence  $u\in \E_{\theta}^*$ and its respective comtrace $\mathbf{t}=[u]\in \E_{\theta}^*/\!\equiv$, we define the relational structures $S_{\mathbf{t}}$ as: \\
\mbox{\hspace{3.5cm}}$\displaystyle S_{\mathbf{t}}  = \bigl( \Sigma_{\mathbf{t}}, \prec_{\mathbf{t}} ,\sqsubset_{\mathbf{t}} \bigr)\df \bigl( \Sigma_{u}, \prec_u ,\sqsubset_u \bigr)^\lozenge$.\EOD
\label{def:s2sos}
\end{definition}

The relational structure $S_{\mathbf{t}}$ is called the \emph{so-structure defined by the comtrace} $\mathbf{t}=[u]$, where $\Sigma_{\mathbf{t}}$, $\prec_{\mathbf{t}}$  and $\sqsubset_{\mathbf{t}}$ are used to denote the event occurrence set, causality relation and weak causality relation induced by the comtrace $\mathbf{t}$ respectively.  The following nontrivial theorem and its corollary justifies the name by showing that  step sequences in a comtrace $\mathbf{t}$ are exactly stratified extension of the so-structure $S_\mathbf{t}$, and that $S_{\mathbf{t}}$ is uniquely defined for the comtrace $\mathbf{t}$ regardless of the choice of $u\in \mathbf{t}$. 

\begin{theorem}[\cite{JK95}]
For each $\mathbf{t}\in \quotient{E^*}{\equiv_{\theta}}$, the relational structure $S_{\mathbf{t}}$ is a so-structure and $ext\bigl(S_{\mathbf{t}}\bigr) = \bset{ \lhd_u \mid u \in \mathbf{t} }$. \END
\label{theo:com2sos}
\end{theorem}

\begin{corollary} For all $\mathbf{t},\mathbf{q}\in \quotient{E^*}{\equiv_{\theta}}$,
\begin{enumerate}
\item $\mathbf{t}=\mathbf{q} \implies S_{\mathbf{t}} = S_{\mathbf{q}}$
\item $ S_{\mathbf{t}}  = \bigl( \Sigma_{\mathbf{t}}, \prec_{\mathbf{t}} ,\sqsubset_{\mathbf{t}} \bigr) = \left( \Sigma_{\mathbf{t}}, \bigcap_{w\in \mathbf{t}}\lhd_{w} ,\bigcap_{w\in \mathbf{t}}\lhd_{w}^{\frown} \right)$ 
\qed
\end{enumerate}
\label{cor:com2sos}
\end{corollary}

\section{Comtraces as Labeled Stratified Order Structures \label{sec:lsos-comtrace}}
Even though \theoref{com2sos} shows that each comtrace can be represented uniquely by a labeled so-structure, it does not give us an explicit definition of how these labeled so-structures look like. In this section, we will give an exact definition of labeled so-structures that represent comtraces. To provide us with more intuition, we first recall how Mazurkiewicz traces can be characterized as labeled posets.

A \emph{trace concurrent alphabet}  is  a pair $(E,ind)$, where
$ind$ is a symmetric irreflexive binary relation on the finite set $E$. A \emph{trace congruence} $\equiv_{ind}$ can then be defined as the smallest  equivalence relation such that for all sequences $uabv, ubav \in E^*$, if $(a,b) \in \textit{ind}$, then $uabv \equiv_{ind} ubav$. The elements of $\quotient{E^{*}}{\equiv_{ind}}$ are called \emph{traces}.  

Traces can also be defined alternatively as posets whose elements are labeled with symbols of a concurrent alphabet $(E,ind)$ satisfying certain conditions.  

Given a binary relation $R\subseteq X$, the \emph{covering relation} of $R$ is defined as $\imm{R}\df \set{(x,y)\mid x\; R\; y \wedge \neg \exists z,\; x\; R\; z\; R \;y }$. An alternative definition of Mazurkiewicz trace is:

\begin{definition}[cf. \cite{TW02}] A trace over a concurrent alphabet $(E,ind)$ is a finite labeled poset $(X,\PO,\lambda)$, where $\lambda:X\rightarrow E$ is a labeling function, such that for all $\alpha,\beta \in X$,
\begin{enumerate}
\item $\alpha \imm{\PO} \beta \implies (\lambda(\alpha),\lambda(\beta))\not \in ind$, and
\item $(\lambda(\alpha),\lambda(\beta))\not \in ind \implies \alpha \PO \beta \vee \beta \PO \alpha$. \EOD
\end{enumerate}
\label{def:ltraces}
\end{definition}

A trace in this definition is only identified unique up to  \emph{label-preserving isomorphism}. The first condition says that immediately causally related event occurrences  must be labeled with dependent events. The second condition ensures that any two event occurrences with dependent labels must be causally related.  The first condition is particularly important since two immediately causally related event occurrences will occur next to each other in at least one of its linear extensions. This is the key to relate  \defref{ltraces} with quotient monoid definition of traces. Thus, we would like to establish a similar relationship for  comtraces. An immediate technical difficulty  is that weak causality might be cyclic, so the notion of ``immediate weak causality''  does not make sense. However, we can still deal with cycles of a so-structure by taking advantage of  the following simple fact: \emph{the weak causality relation is a strict preorder}.

Let $S=(X,\PO,\sq)$ be a so-structure. We define the relation $\eqrel{\sq}\subseteq X\times X$ as\\
$\mbox{\hspace{3cm}}\alpha \eqrel{\sq} \beta \iffdf \alpha=\beta \;\vee\; \bigl(\alpha\sq \beta \wedge \beta \sq \alpha\bigr)$

Since $\sq$ is a strict preorder, it follows that $\eqrel{\sq}$ is an equivalence relation. The relation $\eqrel{\sq}$ will be called the \emph{$\sq$-cycle equivalence relation} and an element of the quotient set $\quotient{X}{\eqrel{\sq}}$ will be called a \emph{$\sq$-cycle equivalence class}.  We then define the following binary relations $\widehat{\PO}$ and $\widehat{\sq}$ on  the quotient set $\quotient{X}{\eqrel{\sq}}$ as
\begin{align}
[\alpha] \widehat{\PO} [\beta] \iffdf ([\alpha]\times[\beta])\;\cap \PO \not=\emptyset
\hspace*{3mm}\text{ and }\hspace*{3mm}
[\alpha] \widehat{\sq} [\beta] \iffdf ([\alpha]\times[\beta])\;\cap \sq \not=\emptyset \label{qsos}
\end{align}

Using this quotient construction,  every so-structure, whose weak causality relation might be cyclic, can be uniquely represented by an \emph{acyclic} quotient so-structure. 
\begin{proposition} The relational structure $\quotient{S}{\eqrel{\sq}} \df (\quotient{X}{\eqrel{\sq}},\widehat{\PO},\widehat{\sq})$ is a so-structure, the relation $\widehat{\sq}$ is a partial order, and for all $x,y\in X$,
\begin{enumerate}
\item $\alpha \PO \beta\iff [\alpha] \widehat{\PO} [\beta]$ 
\item $\alpha \sq \beta \iff [\alpha] \widehat{\sq} [\beta] \vee (\alpha\not=\beta \wedge [\alpha]=[\beta])$  
\end{enumerate}
\end{proposition}
\begin{proof} Follows from \defref{sos}. \qed
\end{proof}

Using (\ref{qsos}) and \theoref{SzpStrat}, it is not hard to prove the following simple yet useful properties of $\sq$-cycle equivalence classes.

\begin{proposition} Let $S = (X,\PO,\sq)$ be a so-structure. We use $u$ and $v$  to denote some step sequences over $\PSB{X}$. Then for all $\alpha,\beta \in X$,
	\begin{enumerate}
	\item $[\alpha]=[\beta] \iff  \forall \lhd \in ext(S),\; \alpha \simeq_{\lhd} \beta$
	\item $\exists\lhd \in ext(S),\; \Omega_{\lhd} = u[\alpha]v$
	\item $[\alpha] \imm{\hat{\sq}} [\beta] \implies \exists\lhd \in ext(S),\;\Omega_{\lhd} = u[\alpha][\beta]v$ \qed
	\end{enumerate}
\label{prop:covlsos}
\end{proposition}

Each $\sq$-cycle equivalence class is what Juh\'as, Lorenz and Mauser called a \emph{synchronous step} \cite{JLM06,JLM08}. They also used equivalence classes to capture synchronous steps but only for the special class of \emph{synchronous closed} so-structures, where $(\sq\setminus \PO)\cup id_{X}$ is an equivalence relation. We extend their ideas by using $\sq$-cycle equivalence classes to capture what we will call  \emph{non-serializable sets} in arbitrary so-structures. The name is justified in assertion (1) of \propref{covlsos} stating that two elements belong to the same non-serializable set of a so-structure $S$ iff they must be executed simultaneously in every stratified extension of $S$. Furthermore, we show in assertion (2) that all elements of a non-serializable set must occur together as a single step in at least one stratified extension of $S$. Assertion (3) gives a sufficient condition for two non-serializable sets to occur as consecutive steps in at least one stratified extension of $S$.

Before we proceed to define comtrace using labeled so-structure, we need to define  \textit{label-preserving isomorphisms} for labeled so-structures more formally. A tuple $T = (X,P,Q,\lambda)$ is a \textit{labeled relational structure} iff $(X,P,Q)$ is a relational structure and $\lambda$ is a function with domain $X$. If $(X,P,Q)$ is a so-structure, then $T$ is a \textit{labeled so-structure}.

\begin{definition}[label-preserving isomorphism] Given two labeled relational structures $T_{1}=(X_{1},P_{1},Q_{1},\lambda_{1})$ and $T_{2}=(X_{2},P_{2},Q_{2},\lambda_{2})$, we write $T_{1}\cong T_{2}$ to denote that $T_1$ and $T_2$ are  \emph{label-preserving isomorphic} (\emph{lp-isomorphic}). In other words, there is a bijection $f: X_{1} \rightarrow X_{2}$ such that for all $\alpha,\beta \in  X_{1}$, 
\begin{enumerate}
\item $(\alpha,\beta) \in P_{1} \iff (f(\alpha),f(\beta))\in P_{2}$
\item $(\alpha,\beta) \in Q_{1} \iff (f(\alpha),f(\beta))\in Q_{2}$
\item $\lambda_{1}(\alpha) = \lambda_{2}(f(\alpha))$
\end{enumerate}
Such function $f$ is called a \emph{label-preserving isomorphism} (\emph{lp-isomorphism}). \EOD
\end{definition}

Note that all notations, definitions and results for so-structures are applicable to labeled so-structures. We also write $[T]$ or $\seq{X,P,Q,\lambda}$ to denote the lp-isomorphic class of a labeled relational structure $T=(X,P,Q,\lambda)$. We  will not distinguish  an lp-isomorphic class $\seq{T}$ with a single labeled relational structure $T$ when it does not cause ambiguity.

We are now ready to give an alternative definition for comtraces. To avoid confusion with the comtrace notion by Janicki and Koutny in \cite{JK95}, we will use the term \emph{lsos-comtrace} to denote a comtrace defined using our definition.

\begin{definition}[lsos-comtrace] Given a comtrace alphabet $\theta=(E,sim,ser)$, a \emph{lsos-comtrace} over $\theta$ is (an lp-isomorphic class of) a finite labeled so-structure $\seq{X,\PO,\sq,\lambda}$ such that $\lambda:X\rightarrow E$  and for all $\alpha,\beta \in X$,
\begin{enumerate}
\item[] \textsf{LC1:\mbox{\hspace{5mm}}} $[\alpha] (\imm{\hat{\sq}}\cap \hat{\PO}) [\beta] \implies \lambda([\alpha])\times \lambda([\beta]) \nsubseteq ser$
\item[] \textsf{LC2:\mbox{\hspace{5mm}}} $[\alpha] (\imm{\hat{\sq}}\setminus \hat{\PO}) [\beta] \implies \lambda([\beta])\times \lambda([\alpha]) \nsubseteq ser$
\item[] \textsf{LC3:\mbox{\hspace{5mm}}} $\forall A,B \in \PSB{[\alpha]},\; A\cup B = [\alpha] \implies \lambda(A)\times \lambda(B) \not \subseteq ser$
\item[] \textsf{LC4:\mbox{\hspace{5mm}}} $(\lambda(\alpha),\lambda(\beta))\not \in ser \implies \alpha \PO \beta \vee \beta \sq \alpha$
\item[] \textsf{LC5:\mbox{\hspace{5mm}}} $(\lambda(\alpha),\lambda(\beta))\not\in sim \implies \alpha\PO \beta \vee \beta \PO\alpha$ 
\end{enumerate}
We write $\LCT(\theta)$ to denote the class of all lsos-comtraces over $\theta$. \EOD
\label{def:lcomtrace}
\end{definition}

\begin{example}
Let $E=\set{a,b,c}$, $sim = \bset{\set{a,b},\set{a,c},\set{b,c}}$ and $ser=\set{(a,b),(b,a),$ $(a,c)}$. Then we have
$\E=\set{\{a\},\{b\},\{c\},\set{b,c}}$. The lp-isomorphic class  of the labeled so-structure $T=(X,\PO,\sq,\lambda)$ depicted in \figref{f1} (the dotted edges denote $\sq$ relation and the solid edges denote both $\PO$ and $\sq$ relations) is a lsos-comtrace. The graph in \figref{f2} represents the labeled quotient  so-structure $\quotient{T}{\eqrel{\sq}}=(\quotient{X}{\eqrel{\sq}},\widehat{\PO},\widehat{\sq},{\lambda'})$ of $T$, where we define ${\lambda'}(A)= \bset{\lambda(x)\mid x\in A}$.
\begin{figure}[ht]
\begin{minipage}{0.45\linewidth}\centering

$\xymatrix@C=3em @R=3em{
*+[o][F-]{a} \ar@{-->}[dr]\ar@/^/[rr] \ar@/_1.5pc/ [ddrr] & 	& *+[o][F-]{c}\ar@/_/ @{-->}[dd]\\
 		 &*+[o][F-]{c} \ar[ur]\ar[dr]		&\\
*+[o][F-]{b} \ar[ur]\ar@/_/[rr]\ar@/^1.5pc/ [uurr]  &	&*+[o][F-]{b}\ar@/_/ @{-->}[uu]
}$
\caption{lsos-comtrace $[T]$}
\label{fig:f1}
\end{minipage}
\begin{minipage}{0.5\linewidth}\centering
$\xymatrix@-1pc@C=2.5em @R=2.7em{
*+[F-]{a} \ar@{-->}[dr]\ar@/^1pc/[drr] & 	&\\
 		 &*+[F-]{c} \ar[r]		&*+[F-]{b,c}\\
*+[F-]{b} \ar[ur]\ar@/_1pc/[urr]  &	&
}$
\caption{the quotient structure $\quotient{T}{\eqrel{\sq}}$ of $T$}
\label{fig:f2}
\end{minipage}
\end{figure}

The lsos-comtrace $[T]$ actually corresponds to the comtrace $[\set{a,b}\set{c}\set{b,c}]$, and we will show this relationship formally in \secref{representation}. 
\EOD
\label{ex:comtrace2}
\end{example}

\begin{remark}
\defref{lcomtrace} can be extended to define \emph{infinite comtrace} as follows. Instead of asking $X$ to be finite, we require a labeled so-structure to be  \emph{initially finite} (cf. \cite{JK97}), i.e., $\bset{\alpha\in X\mid \alpha \sq\beta}$ is finite for all $\beta\in X$. The initially-finiteness not only gives us a sensible interpretation that every event only  causually depends  on finitely many events, but also guarantees that the covering relations of $\hat{\PO}$ and $\hat{\sq}$ are well-defined. \EOD
\end{remark}

Since each lsos-comtrace is defined as a class of lp-isomorphic labeled so-structures, dealing with lsos-comtrace might seem tricky. Fortunately, the \emph{no autoconcurrency} property, i.e., the relation $ser$ is irreflexive, gives us a \emph{canonical} way to enumerate the events of a lsos-comtrace very similar to how the events of a comtrace are enumerated.

Given a step sequence $s=A_{1}\ldots A_{k}$ and any function $f$ defined on $\bigcup_{i=1}^k A_i$, we define $\map(f,s)\df f(A_{1})\ldots f(A_{k})$, i.e., the step sequence derived from $s$ by applying the function $f$ successively on each $A_{i}$. Note that $f(A_{i})$ denotes the \emph{image} of $A_i$ under $f$.

Given a lsos-comtrace $T=\seq{X,\PO,\sq,\lambda}$  over a comtrace alphabet $\theta=(E,sim,ser)$,  a stratified order $\lhd\in ext(T)$ can be seen as a step sequence $\Omega_{\lhd}=A_{1}\ldots A_{k}$.

\begin{proposition}  
\begin{enumerate}
\item For every $i$ ($1\le i \le k$), $|A_{i}|=|\lambda(A_{i})|$
\item $\map(\lambda,\Omega_{\lhd}) = \lambda(A_{1})\ldots \lambda(A_{k}) \in \E_{\theta}^{*}$. \qed
\end{enumerate}
\label{prop:validss}
\end{proposition}

 \propref{validss} ensures that  $u=\map(\lambda,\Omega_{\lhd})$ is a valid step sequence over $\theta$. Recall that $\h{u}=\h{A_{1}}\ldots \h{A_{k}}$ denotes the enumerated step sequence of $u$ and $\Sigma_{u}$ denotes the set of event occurrences. Define a  \emph{bijection} $\xi_{u}:\Sigma_{u}\rightarrow X$ as\\
\mbox{\hspace{2.7cm}}$\xi_{u}(\alpha)  = x \iffdf \alpha \in \h{A_{i}} \;\wedge\; x\in A_{i} \;\wedge\; \lambda(x) = \ell(\alpha)$\smallskip

By \propref{validss}, the function $\xi_{u}$ is well-defined.  Moreover, we can show that $\xi_{u}$ is uniquely determined by $T$ regardless of the choice of $\lhd\in ext(T)$.

\begin{proposition} Given $\lhd_{1},\lhd_{2}\in ext(T)$, let $v  = \map(\lambda,\Omega_{\lhd_{1}})$ and $w = \map(\lambda,\Omega_{\lhd_{2}})$. Then $\xi_{v} = \xi_{w}$. \qed
\label{prop:uniquexi}
\end{proposition}

Henceforth, we will ignore subscripts and reserve  the notation $\xi$  to denote the kind of mappings as defined above. We then define the \emph{enumerated so-structure} of $T$ to be the labeled so-structure $T_{0}=(\Sigma,\PO_{0},\sq_{0},\ell)$, where $\Sigma = \Sigma_{u}$ for $u=\map(\lambda,\Omega_{\lhd})$ and $\lhd\in ext(T)$;  and the relations $\PO_{0},\sq_{0}\subseteq \Sigma \times \Sigma$ are defined as\\
\mbox{\hspace{1.6cm}}$\alpha \PO_{0} \beta \iffdf \xi(\alpha) \PO \xi(\beta)
\andspace
\alpha \sq_{0} \beta \iffdf \xi(\alpha) \sq \xi(\beta)
$

Clearly, the  enumerated so-structure $T_0$ can be uniquely determined from $T$ using the preceding definition.  From our construction, we can easily show the following important relationships:

\begin{proposition} 
\begin{enumerate}
\item $T_{0}$ and $T$ are  lp-isomorphic under the mapping $\xi$. 
\item The labeled so-structures  $(\Sigma,\lhd_u,\lhd_u^{\frown},\ell)$ and $(X,\lhd,\lhd^{\frown},\lambda)$ are lp-isomorphic under the mapping $\xi$ and $\lhd_u\in ext(T_{0})$. \qed 
\end{enumerate} 
\label{prop:isoExt}
\end{proposition}
\begin{minipage}{7.2cm}
\hspace{0.5cm}In other words, the mapping $\xi:\Sigma\rightarrow X$ plays 
the role of both the lp-isomorphism from $T_{0}$ to $T$ and the lp-isomorphism from the stratified extension $(\Sigma,\lhd_{u})$ of $T_{0}$ to the stratified extension $(X,\lhd)$ of $T$. These relationships can be best captured using the commutative diagram on the right.\smallskip
\end{minipage}
\begin{minipage}{4.5cm}\centering\small\vspace{-3mm}
$\xymatrix{
 (\Sigma,\PO_{0},\sq_{0},\ell) \ar@{->}[r]_{\displaystyle \xi } \ar@{^{(}->}[d]_{\displaystyle id_{\Sigma}}
      & (X,\PO,\sq,\lambda) \ar@{^{(}->}[d]^{\displaystyle id_{X}}
   \\ 
 (\Sigma,\lhd_{u},\lhd_{u}^{\frown},\ell) \ar@{->}[r]^{\displaystyle \xi} &  (X,\lhd,\lhd^{\frown},\lambda)}$
 \end{minipage}

We can even observe further that two lsos-comtraces are identical if and only if they define the same enumerated so-structure. Henceforth, we will call an enumerated so-structure defined by a lsos-comtrace $T$ \emph{the canonical representation} of $T$. 

Recently, inspired by the dependency graph notion for Mazurkiewicz traces (cf. \cite[Chapter 2]{DR}), Kleijn and Koutny claimed without proof that their \emph{combined dependency graph} notion is another alternative way to define comtraces \cite{KK08}. In \secref{representation}, we will give a detailed proof of their claim.

\begin{definition}[combined dependency graph \cite{KK08}] 
Given an comtrace alphabet $\theta=(E,ser,sim)$, a \emph{combined dependency graph} (\emph{cd-graph}) over $\theta$ is (a lp-isomorphic class of) a finite labeled relational structure $D=\seq{X,\cau,\wcau,\lambda}$ such that $\lambda:X\rightarrow E$, the relations $\cau,\wcau$ are irreflexive, $D^{\lozenge}$ is a so-structure, and for all $\alpha,\beta\in X$,\smallskip

\textsf{CD1:\mbox{\hspace{5mm}}} $(\lambda(\alpha),\lambda(\beta))\not\in sim \implies \alpha \cau \beta \vee \beta \cau \alpha$

\textsf{CD2:\mbox{\hspace{5mm}}} $(\lambda(\alpha),\lambda(\beta))\not\in ser \implies \alpha \cau \beta \vee \beta \wcau \alpha$

\textsf{CD3:\mbox{\hspace{5mm}}} $\alpha \cau \beta \implies (\lambda(\alpha),\lambda(\beta))\not\in ser$

\textsf{CD4:\mbox{\hspace{5mm}}} $\alpha \wcau \beta \implies (\lambda(\beta),\lambda(\alpha))\not\in ser$\smallskip\\
We will write $\DCD(\theta)$ to denote the class of all cd-graphs over $\theta$. \EOD
\label{def:comdag}
\end{definition}

Cd-graphs can be seen as reduced graph-theoretic representations for lsos-comtraces, where some arcs that can be recovered using $\lozenge$-closure are omitted. It is interesting to observe that the non-serializable sets of a cd-graph are exactly the \emph{strongly connected components} of the directed graph $(X,\wcau)$ and can easily be found in time $O(|X|+|\wcau|)$ using any standard algorithm (cf. \cite[Section 22.5]{CLR}). \smallskip\\
\begin{minipage}{8.2cm}\begin{remark} Cd-graphs were called \emph{dependence comdags} in \cite{KK08}. But this name could be misleading since the directed graph $(X,\wcau)$ is not necessarily acyclic. For example, the graph on the right is the cd-graph that corresponds to  the lsos-comtrace from \figref{f1}, but it is not acyclic. (Here, we use the dotted edges to denote $\wcau$ and the solid edges to denote \emph{only} $\cau$.) Thus, we use the name ``combined dependency graph'' instead. \EOD\end{remark}
\end{minipage}
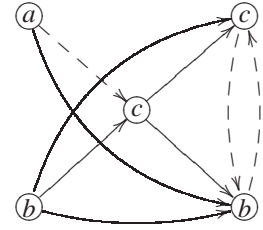
\begin{minipage}{4cm}\centering
$\xymatrix@C=3em @R=2.5em{
*+[o][F-]{a} \ar@{-->}[dr] \ar@/_1.5pc/ [ddrr] & 	& *+[o][F-]{c}\ar@/_/ @{-->}[dd]\\
 		 &*+[o][F-]{c} \ar[ur]\ar[dr]		&\\
*+[o][F-]{b} \ar[ur]\ar@/_/[rr]\ar@/^1.5pc/ [uurr]  &	&*+[o][F-]{b}\ar@/_/ @{-->}[uu]
}$
\end{minipage}

\section{Representation Theorems \label{sec:representation}}
This section contains the main technical contribution of this paper by showing that for a given comtrace alphabet $\theta$, $\E^*/\!\!\equiv_{\theta}$, $\LCT(\theta)$ and $\DCD(\theta)$ are three equivalent ways of talking about the same class of objects. We will next prove the first representation theorem which establishes the representation mappings between $\E^*/\!\!\equiv_{\theta}$ and $\LCT(\theta)$.

\subsection{Representation Theorem for Comtraces and lsos-Comtraces}
\begin{proposition}
Let $S_0=(X,\PO_0,\sq_0)$ and $S_1=(X,\PO_1,\sq_1)$ be stratified order structures such that $ext(S_0)\subseteq ext(S_1)$. Then $S_1\subseteq S_0$.
\label{prop:stratsubset}
\end{proposition}
\begin{proof} Follows from \theoref{SzpStrat}. \qed
\end{proof}

For the next two lemmata, we let $T$ be a lsos-comtrace over a comtrace alphabet $\theta=(E,sim,ser)$. Let $T_{0}=(\Sigma_{u},\PO_{0},\sq_{0},\ell)$ be the canonical representation of $T$. Let $\lhd_{0}\in ext(T_0)$ and $u=\map(\ell,\Omega_{\lhd_{0}})$. Since $u$ is a valid step sequence in $\E^*$ (by \propref{validss}), we can construct $S_{[u]} = (\Sigma_{u},\PO_{[u]},\sq_{[u]})$ from \defref{s2sos}. Our goal is to show that the stratified order $S_{[u]}$ defined by the comtrace $[u]$ is exactly $(\Sigma_{u},\PO_{0},\sq_{0})$.

\begin{lemma} $S_{[u]} \subseteq (\Sigma_{u},\PO_{0},\sq_{0})$.
\label{lem:l1}
\end{lemma}
\begin{proof} By \propref{so-cl}, to show $S_{[u]} = (\Sigma_{u},\PO_{u},\sq_{u})^{\lozenge} \subseteq (\Sigma_{u},\PO_{0},\sq_{0})$, it suffices to show that $(\Sigma_{u},\PO_{u},\sq_{u}) \subseteq (\Sigma_{u},\PO_{0},\sq_{0})$. Since $T_0$ is the canonical representation of $T$, it is important to observe that $\lhd_0=\lhd_u$.

($\PO_{u}\subseteq \PO_{0}$): Assume $\alpha\PO_{u}\beta$. Then from \defref{s2inv}, $\alpha\lhd_u\beta \wedge (\ell(\alpha),\ell(\beta))\notin ser$. Since $(\ell(\alpha),\ell(\beta))\notin ser$, it follows from \defref{lcomtrace} that  $\alpha \PO_{0} \beta$ or $\beta\sq_{0} \alpha$.  Suppose for a contradiction that $\beta\sq_{0} \alpha$, then by \theoref{SzpStrat},  $\forall \lhd \in ext(T_{0}),\; \beta \lhd^{\frown}\alpha$. But since we assume that $\lhd_{0}\in ext(T_0)$, it follows that $\lhd_{u}\in ext(T_{0})$ and $\alpha \lhd_{u} \beta$, a contradiction. Hence, we have shown $\alpha \PO_{0} \beta$.

($\sq_{u}\subseteq \sq_{0}$):  Can be shown in a similar way.
\qed
\end{proof}

\begin{lemma} $S_{[u]} \supseteq (\Sigma_{u},\PO_{0},\sq_{0})$.
\label{lem:l2}
\end{lemma}

In this proof, we will include subscripts for equivalence classes to avoid confusing the elements from quotient set $\quotient{\Sigma_{u}}{\equiv_{\sq_{0}}}$ with the elements from the quotient comtrace monoid $\quotient{\E^*}{\equiv_{\theta}}$. In other words, we write $[\alpha]_{\equiv_{\sq_{0}}}$ to denote an element of the quotient set $\quotient{\Sigma_{u}}{\equiv_{\sq_{0}}}$, and write $[u]_{\theta}$ to denote the comtrace generated by $u$.

\begin{proof}[of \lemref{l2}] Let $S'=(\Sigma_{u},\PO_{0},\sq_{0})$. To show $S_{[u]} \supseteq S'$, by \propref{stratsubset}, it suffices to show $ext(S_{[u]} )\subseteq ext(S')$. From \theoref{com2sos}, we know that $ext(S_{[u]_{\theta}}) = \set{\lhd_w\mid w\in [u]_{\theta}}$. Thus we only need to show that for all $w\in [u]_{\theta}$, $\lhd_w \in ext(S')$.

We observe that from $u$, by \defref{commonoid}, we can generate all the step sequences in the comtrace $[u]_{\theta}$ in stages using the following recursive definition:
\begin{align*}
D^0(u)&\df \set{u}\\
D^n(u)&\df \set{w\mid w\in D^{n-1}(u)\;\vee\; \exists v\in D^{n-1}(u),\ (\ v\eqa_{\theta} w\; \vee\; v\eqa_{\theta}^{-1} w)} 
\end{align*}

Since the set $[u]_{\theta}$ is finite, $[u]_{\theta}=D^n(u)$ for some stage $n\ge 0$. For the rest of the proof, we will prove by induction on $n$ that for all $n\in \mathbb{N}$, if $w\in D^n(u)$ then $\lhd_w\in ext(S)$. \\
\textbf{Base case:} When $n=0$, $D^0(u)=\set{u}$. Since $\lhd_0 \in ext(T)$, it follows from \propref{isoExt} that $\lhd_u \in ext(S')$. \\
\textbf{Inductive case:} When $n>0$, let $w$ be an element of $D^n(u)$. Then either $w\in D^{n-1}(u)$ or $w\in (D^n(u)\setminus D^{n-1}(u))$. For the former case, by inductive hypothesis, $\lhd_w\in ext(S')$. For the latter case, there must be some element $v\in D^{n-1}(u)$ such that $v\eqa_{\theta} w$ or  $v\eqa_{\theta}^{-1} w$. By induction hypothesis, we already known $\lhd_v\in ext(S')$. We want to show that $\lhd_w \in ext(S')$. There are two cases to consider:\smallskip\\
\textit{\textbf{Case (i):}} \\
When $v\eqa_{\theta} w$, by \defref{commonoid}, there are some $y,z\in E_{\theta}^*$ and steps $A,B,C\in \E$ such that $v=yAz$ and $w=yBCz$ where $A$, $B$, $C$ satisfy $B\cap C =\emptyset$ and $B\cup C = A$ and $B\times C\subseteq ser$. Let $\h{v}=\h{y}\h{A}\h{z}$ and $\h{w}=\h{y}\h{B}\,\h{C}\h{z}$ be enumerated step sequences of $v$ and $w$ respectively. 

Suppose for a contradiction that $\lhd_w \not\in ext(S')$. By \defref{extsos}, there are $\alpha\in \h{C}$  and $\beta \in \h{B}$ such that $\alpha \sq_{0} \beta$. We now consider the quotient set $\quotient{\h{A}}{\equiv_{\sq_{0}}}$. By  \propref{covlsos} (1), $\quotient{\h{A}}{\equiv_{\sq_{0}}} \subseteq \quotient{\Sigma_{u}}{\equiv_{\sq_{0}}}$. 
Since $\alpha \sq_{0} \beta$, it follows that $[\alpha]_{\equiv_{\sq_{0}}}\hat{\sq}_{0} [\beta]_{\equiv_{\sq_{0}}}$. Thus, from the fact that $\hat{\sq}_{0}$ is partial order, there must exists a chain
\begin{align}
[\alpha]_{\equiv_{\sq_{0}}}=\; [\gamma_1]_{\equiv_{\sq_{0}}} \;\imm{\hat{\sq}_{0}}\; [\gamma_2]_{\equiv_{\sq_{0}}}\;\imm{\hat{\sq}_{0}}\; \ldots\;\imm{\hat{\sq}_{0}}\;[\gamma_k]_{\equiv_{\sq_{0}}}\;= [\beta]_{\equiv_{\sq_{0}}} 
\label{eq:chain1}
\end{align}

Then by \theoref{SzpStrat} and the fact that $\lhd_v\in ext(S')$, we know that $\gamma_i \in \h{A}$ for all $i$. In other words, since the chain \eref{chain1} implies that every $\gamma_i$ must always occur between $\alpha$ and $\beta$ in all stratified extensions of $S'$ and $\alpha,\beta \in \h{A}$, we also have $\gamma_i \in \h{A}$. Hence, by \propref{covlsos} (1), we have $[\gamma_i]_{\equiv_{\sq_{0}}}\subseteq \h{A}$ for all $i$, $1\le i \le k$. Also from \textsf{LC3} of \defref{lcomtrace} and that $B\times C \subseteq ser$, we know that for each $\gamma_i$, either $[\gamma_i]_{\equiv_{\sq_{0}}}\subseteq \h{B}$ or $[\gamma_i]_{\equiv_{\sq_{0}}}\subseteq \h{C}$. Now we note  that the first element on the chain $[\gamma_1]_{\equiv_{\sq_{0}}}=[\alpha]_{\equiv_{\sq_{0}}}\subseteq \h{C}$ and the last element on the chain $[\gamma_k]_{\equiv_{\sq_{0}}}=[\beta]_{\equiv_{\sq_{0}}}\subseteq \h{B}$. Thus, there exist two consecutive elements $[\gamma_i]_{\equiv_{\sq_{0}}}$ and $[\gamma_{i+1}]_{\equiv_{\sq_{0}}}$ on the chain such that $[\gamma_i]_{\equiv_{\sq_{0}}}\subseteq \h{C}$ and $[\gamma_{i+1}]_{\equiv_{\sq_{0}}}\subseteq \h{B}$. But then it follows that
\begin{itemize}
 \item[(a)] $[\gamma_{i+1}]_{\equiv_{\sq_{0}}}\times [\gamma_i]_{\equiv_{\sq_{0}}}\subseteq ser$ and $[\gamma_i]_{\equiv_{\sq_{0}}} \imm{\hat{\sq}_{0}} [\gamma_{i+1}]_{\equiv_{\sq_{0}}}$
 \item[(b)] $\neg \bigl([\gamma_i]_{\equiv_{\sq_{0}}} \hat{\PO}_{0} [\gamma_{i+1}]_{\equiv_{\sq_{0}}}\bigr)$ since $\lhd_v\in ext(S')$ and $\gamma_i\frown_{\lhd_v}\gamma_{i+1}$
\end{itemize}
These contradict \textsf{LC2} of \defref{lcomtrace} since $T_0$ is a lsos-comtrace.\smallskip\\
\textit{\textbf{Case (ii):}} \\
When $v\eqa_{\theta}^{-1} w$, by \defref{commonoid}, there are some $y,z\in E_{\theta}^*$ and steps $A,B,C\in \E$ such that $v=yBCz$ and $w=yAz$  where $A$, $B$, $C$ satisfy $B\cap C =\emptyset$ and $B\cup C = A$ and $B\times C\subseteq ser$. Let $\h{v}=\h{y}\h{B}\,\h{C}\h{z}$ and $\h{w}=\h{y}\h{A}\h{z}$ be enumerated step sequences of $v$ and $w$ respectively. 

Suppose for a contradiction that $\lhd_w \not\in ext(S')$. By \defref{extsos}, there are $\alpha\in \h{B}$  and $\beta \in \h{C}$ such that $\alpha \PO_{0} \beta$. By  \propref{covlsos} (1), $\quotient{\h{A}}{\equiv_{\sq_{0}}} \subseteq \quotient{\Sigma_{u}}{\equiv_{\sq_{0}}}$. Thus, using a dual argument to the proof of Case (i), we can build a chain
\begin{align}
[\alpha]_{\equiv_{\sq_{0}}}=\; [\gamma_1]_{\equiv_{\sq_{0}}} \;\imm{\hat{\sq}_{0}}\; [\gamma_2]_{\equiv_{\sq_{0}}}\;\imm{\hat{\sq}_{0}}\; \ldots\;\imm{\hat{\sq}_{0}}\;[\gamma_k]_{\equiv_{\sq_{0}}}\;= [\beta]_{\equiv_{\sq_{0}}} 
\end{align}
We then argue that there are two consecutive elements on the chain such that $[\gamma_i]_{\equiv_{\sq_{0}}}\subseteq \h{B}$ and $[\gamma_{i+1}]_{\equiv_{\sq_{0}}}\subseteq \h{C}$, which implies 
\begin{enumerate}
 \item[(a)] $[\gamma_{i}]_{\equiv_{\sq_{0}}}\times [\gamma_{i+1}]_{\equiv_{\sq_{0}}}\subseteq ser$ and $[\gamma_i]_{\equiv_{\sq_{0}}} \imm{\hat{\sq}_{0}} [\gamma_{i+1}]_{\equiv_{\sq_{0}}}$
 \item[(b)] $[\gamma_i]_{\equiv_{\sq_{0}}} \hat{\PO}_{0} [\gamma_{i+1}]_{\equiv_{\sq_{0}}}$ since $\lhd_v\in ext(S')$ and $\gamma_i \lhd_v \gamma_{i+1}$
\end{enumerate}
These contradict  \textsf{LC1} of \defref{lcomtrace}. \qed

\end{proof}

We also need to show that the labeled so-structure defined from each comtrace is indeed a lsos-comtrace. In other words, we need to show the following lemma.

\begin{lemma}  Let $\theta=(E,sim,ser)$ be a comtrace alphabet. Given a step sequence  $u \in \E_{\theta}^{*}$, the lp-isomorphic class $\seq{\Sigma_{[u]},\PO_{[u]},\sq_{[u]},\ell}$ is a lsos-comtrace over $\theta$. \qed
\label{lem:l3}
\end{lemma}

The proof of this lemma is straightforward by checking that $\seq{\Sigma_{[u]},\PO_{[u]},\sq_{[u]},\ell}$ satisfies all conditions \textsf{LC1}--\textsf{LC5}.

\begin{definition}[representation mappings $\cl$ and $\lc$] Let $\theta$ be a comtrace alphabet. 
\begin{enumerate}
 \item The mapping  $\cl:\quotient{\E_\theta^*}{\equiv_{\theta}} \rightarrow \LCT(\theta)$ is defined as\smallskip\\
 \mbox{\hspace{4cm}}$\cl(\mathbf{t}) \df \seq{\Sigma_{\mathbf{t}},\PO_{\mathbf{t}},\sq_{\mathbf{t}},\ell},$\smallskip\\
where the function $\ell:\Sigma_s\rightarrow E$ is defined in \secref{steps} and $S_{\mathbf{t}}=(\Sigma_{\mathbf{t}},\PO_{\mathbf{t}},\sq_{\mathbf{t}})$ is the so-structure defined by the comtrace $\mathbf{t}$ from \defref{s2sos}. 
 \item The mapping  $\lc:\LCT(\theta)\rightarrow \quotient{\E_\theta^*}{\equiv_{\theta}}$ is defined as\smallskip\\
\mbox{\hspace{1.5cm}} $\lc\bigl((X,\PO,\sq,\lambda)\bigr) \df \Bset{map(\lambda,\Omega_\lhd)\mid \lhd \in ext\bigl((X,\PO,\sq)\bigr)}.$
 \EOD
\end{enumerate}
\label{def:repmaps}
\end{definition}

Intuitively, the mapping $\cl$ is used to convert a comtrace to lsos-comtrace while the mapping $\lc$ is used to transform a lsos-comtrace into a comtrace. The fact that $\cl$ and $\lc$ are valid representation mappings for $\quotient{\E_\theta^*}{\equiv_{\theta}}$ and  $\LCT(\theta)$ will be  shown in the following theorem.

\begin{theorem}[The 1$^{st}$ Representation Theorem]  Let $\theta$ be a comtrace alphabet.
\begin{enumerate}
\item For every $\mathbf{t}\in \quotient{\E_\theta^*}{\equiv_{\theta}}$, $\lc\circ \cl(\mathbf{t}) = \mathbf{t}$.
\item For every $T\in \LCT(\theta)$, $\cl\circ \lc(T) = T$.
\end{enumerate}
\label{theo:rep}
\end{theorem}

\begin{proof}\textbf{1. } The fact that $ran(\cl)\subseteq \LCT(\theta)$ follows from \lemref{l3}. Now for a given $\mathbf{t}\in \quotient{\E_\theta^*}{\equiv_{\theta}}$,  we have $\cl(\mathbf{t})=(\Sigma_{\mathbf{t}},\PO_{\mathbf{t}},\sq_{\mathbf{t}},\ell)$. Thus, it follows that
\begin{align*}
\lc(\cl(t))	& = \bset{map(\ell,\Omega_\lhd)\mid \lhd \in ext(S_\mathbf{t})} 	& \\
		& = \bset{map(\ell,\Omega_\lhd)\mid \lhd \in \set{\lhd_s\mid s\in \mathbf{t}}}	&\ttcomment{\text{by \theoref{com2sos}}}\\
		& = \bset{map(\ell,\Omega_{\lhd_{s}})\mid s\in \mathbf{t}}	  = \mathbf{t}	&
\end{align*} 

\textbf{2. }  Assume $T_{0}=(\Sigma,\PO_{0},\sq_{0},\ell)$ is the canonical representation of $T$. Observe that since $T_0\cong T$, we have $\bset{map(\ell,\Omega_\lhd)\mid \lhd \in ext(T_0)} = \bset{map(\lambda,\Omega_\lhd)\mid \lhd \in ext(T)}.$

Let $\Delta = \bset{map(\ell,\Omega_\lhd)\mid \lhd \in ext(T_0)}$. We will next show that $\Delta\in \quotient{\E_\theta^*}{\equiv_{\theta}}$ and  $\cl\bigl(\Delta\bigr) = \seq{T_0}$. Fix an arbitrary $u\in\Delta$, from Lemmas \ref{lem:l1} and \ref{lem:l2},  $S_{[u]} = (\Sigma,\PO_{0},\sq_{0})$. From \theoref{com2sos}, $\Delta =  \bset{map(\ell,\Omega_\lhd)\mid \lhd \in ext(S_{[u]})} = [u]$. And the rest follows. \qed
\end{proof}

The theorem says that the mappings $\cl$ and $\lc$ are inverses of each other and hence are both \emph{bijective}.

\subsection{Representation Theorem for lsos-Comtraces and Combined Dependency Graphs}
Using \theoref{rep}, we are going to show that the \emph{combined dependency graph} notion proposed in \cite{KK08} is another correct alternative definition for comtraces. First we need to define several  representation mappings  that are needed for our proofs.

\begin{definition}[representation mappings $\cd$, $\dl$ and $\ld$] Let $\theta$ be a comtrace alphabet. 
\begin{enumerate}
 \item The mapping  $\cd:\quotient{\E_\theta^*}{\equiv_{\theta}} \rightarrow \DCD(\theta)$ is defined as\smallskip\\
 \mbox{\hspace{3.5cm}} $\cd(\mathbf{t}) \df (\Sigma_{\mathbf{t}},\PO_{u},\sq_{u},\ell),$ \smallskip\\
where $u$ is any step sequence in $\mathbf{t}$ and $\PO_{u}$ and $\sq_{u}$ are defined as in \defref{s2inv}. 
 \item The mapping  $\dl:\DCD(\theta)\rightarrow \LCT(\theta)$ is defined as  $\dl(D) \df D^{\lozenge}$.
 \item The mapping  $\ld:\LCT(\theta)\rightarrow \DCD(\theta)$ is defined as\smallskip\\
 \mbox{\hspace{3.2cm}}   $\ld(T)\df \cd \circ \lc (T) $. \EOD
\end{enumerate}
\label{def:drepmaps}
\end{definition}

Before proceeding futher, we want to make sure that:
\begin{lemma}
\begin{enumerate}
\item $\dl:\DCD(\theta)\rightarrow \LCT(\theta)$ is a well-defined function.
\item $\cd:\quotient{\E_\theta^*}{\equiv_{\theta}} \rightarrow \DCD(\theta)$ is a well-defined function.
\end{enumerate}
\end{lemma}

\begin{proof} \textbf{1. } Given a cd-graph $D_1= \seq{X,\cau_1,\wcau_1,\lambda}\in \DCD(\theta)$, let $T= \seq{X,\PO,\sq,\lambda} = D_1^{\lozenge}$. We know that $T$ is uniquely defined, since by \defref{comdag}, $(X,\PO,\sq)$ is a so-structure, and so-structures are fixed points of $\lozenge$-closure (by \propref{so-cl} (4)). We will next show that $T$ is a lsos-comtrace by verifying the conditions \textsf{LC1}--\textsf{LC5} of \defref{lcomtrace}. Conditions \textsf{LC4} and \textsf{LC5} are exactly \textsf{CD1} and \textsf{CD2}.

\textsf{LC1}: Suppose for contradiction that there exist two distinct non-serializable sets $[\alpha],[\beta]\subset X$ such that $[\alpha](\imm{\hat{\sq}}\cap \hat{\PO})[\beta]$ and $\lambda([\alpha])\times \lambda([\beta]) \subseteq ser$. Clearly, this implies that $\alpha \PO \beta$, and thus by the $\lozenge$-closure definition, $\beta$ is reachable from $\alpha$ on the directed graph $G=(X,\bcau)$, where $\bcau = \cau\cup \wcau$. Now we consider a shortest path $P$\smallskip\\
\mbox{\hspace{3.3cm}}$\alpha = \delta_1 \bcau\delta_2 \bcau\ldots\bcau \delta_{k-1}\bcau \delta_k = \beta$ \smallskip\\
on $G$ that connects $\alpha$ to $\beta$. We will prove by induction on $k\ge 2$ that there exist two consecutive $\delta_i$ and $\delta_{i+1}$ on $P$ such that $\delta_i\in [\alpha]$ and $\delta_{i+1}\in [\beta]$ and $(\lambda(\delta_{i}),\lambda(\delta_{i+1}))\not \in ser$, which contradicts with  $\lambda([\alpha])\times \lambda([\beta]) \subseteq ser$. \\
\textbf{Base case:} when $k=2$, then $\alpha \bcau \beta$. Since $[\alpha](\imm{\hat{\sq}}\cap \hat{\PO})[\beta]$, we have $\alpha \cau \beta$, which by \textsf{CD3} implies $(\lambda(\alpha),\lambda(\beta))\not \in ser$.\\ \textbf{Inductive case:} when $k>2$, we consider  $\delta_1$ and  $\delta_2$. If $\delta_1\in [\alpha]$ and $\delta_2\in [\beta]$, then by $[\alpha](\imm{\hat{\sq}}\cap \hat{\PO})[\beta]$,  we have $\delta_1 \cau \delta_2$, which immediately yields $(\lambda(\delta_{1}),\lambda(\delta_{2}))\not \in ser$. Otherwise, we have $\delta_2\not\in [\alpha]\cup [\beta]$ or $\bset{\delta_1,\delta_2}\subseteq [\alpha]$. For the first case, we get $[\alpha]\hat{\sq}[\delta_2]\hat{\sq}[\beta]$, which contradicts that $[\alpha]\imm{\hat{\sq}}[\beta]$. For the latter case, we can apply induction hypothesis on the path $\delta_2 \bcau\ldots\bcau \delta_{k-1}\bcau \delta_k$.

\textsf{LC2} and \textsf{LC3} can also be shown similarly using ``shortest path'' argument as above. These proofs are easier since we only need to consider paths with edges in $\wcau$. 
 
\textbf{2. } By the proof of \cite[Lemma 4.7]{JK95}, for any two step sequences $t$ and $u$ in $\E_\theta^*$, we have $u \eqb t$ iff $\cd([u]) = \cd([t])$. Thus the mapping $\cd$ is well-defined. \qed
\end{proof}

\begin{lemma} The mapping  $\dl:\DCD(\theta)\rightarrow \LCT(\theta)$ is injective.
\label{lem:dlinj}
\end{lemma}
\begin{proof} Assume that $D_1,D_2 \in \DCD(\theta)$, such that $\dl(D_1)=\dl(D_2) = T= [X,\PO,\sq,\lambda]$. Since $\lozenge$-closure operator does not change the labeling function, we can assume that $D_i= [X,\cau_i,\wcau_i,\lambda]$ and $(X,\cau_i,\wcau_i)^{\lozenge} = (X,\PO,\sq)$. We will next show that $(X,\cau_1,\wcau_1)\subseteq(X,\cau_2,\wcau_2)$.

($\cau_1\;\subseteq\;\cau_2$): Let $\alpha,\beta\in X$ such that $\alpha \cau_1 \beta$. Suppose for a contradiction that $\neg (\alpha \cau_2 \beta)$. Since $\alpha \cau_1 \beta$, by \textsf{CD3},  $(\lambda(\alpha),\lambda(\beta))\not\in ser$. Thus, by \textsf{CD2}, $\beta \wcau_2 \alpha$. But since $(X,\cau_i,\wcau_i)^{\lozenge} = (X,\PO,\sq)$, it follows that $(X,\cau_i,\wcau_i) \subseteq (X,\PO,\sq)$ (by \propref{so-cl}). Thus, $\alpha\PO\beta$ and $\beta \sq \alpha$, a contradiction.

($\wcau_1\;\subseteq\;\wcau_2$): Can be proved similarly.

By reversing the role of $D_1$ and $D_2$, we have $(X,\cau_1,\wcau_1)\supseteq(X,\cau_2,\wcau_2)$. Thus, we conclude $D_1=D_2$.
\qed
\end{proof}

We are now ready to show the following representation theorem which ensures that $\ld$ and $\dl$ are valid representation mappings for $\LCT(\theta)$ and $\DCD(\theta)$.

\begin{theorem}[The 2$^{nd}$ Representation Theorem] Let $\theta$ be a comtrace alphabet.
\begin{enumerate}
\item For every $D\in \DCD(\theta)$, $\ld \circ \dl(D) = D$.
\item For every $T\in \LCT(\theta)$, $\dl \circ \ld(T) = T$.
\end{enumerate}
\label{theo:deprep}
\end{theorem}
\begin{proof}\textbf{1. } Let $D \in \DCD(\theta)$ and let $T = \dl(D)$. Suppose for a contradiction that $E =\ld \circ \dl(D)$ and $E\not = D$. From how $\cl$ is defined, $\cl = \dl\circ \cd$. Thus, it follows that $\dl(E)= T = \dl(D)$. But this contradicts the injectivity of $\dl$ from \lemref{dlinj}.

\textbf{2. } Let $T\in \LCT(\theta)$ and let $D = \ld(T)$. Suppose for a contradiction that $Q = \dl \circ \ld(T)$ and $Q\not= T$. Since $\ld = \cd\circ \lc$, if we let $\mathbf{t}= \lc(T)$, then  $Q = \dl\circ \cd (\mathbf{t})\not= T$. Thus, we have shown that $\mathbf{t}=\lc(T)$ and $\cl(\mathbf{t})=\dl\circ \cd (\mathbf{t})\not= T$, contradicting \theoref{rep} (2).
\qed 
\end{proof}

This theorem shows that lsos-comtraces and cd-graphs are equivalent representations for comtraces.  The main advantage of cd-graph definition is its simplicity while the lsos-comtrace definition is stronger and more convenient to prove properties about labeled so-structures that represent comtraces. 

We do not need to prove another representation theorem for cd-graphs and comtraces since their representation mappings are simply the composition of the representation mappings from Theorems \ref{theo:rep} and \ref{theo:deprep}.

\section{Composition Operators}
Recall for a comtrace monoid $(\E^*/\!\!\equiv_{\theta},\circledast,[\epsilon])$, the comtrace operator $\_\circledast\_$ is defined as $[r]\circledast[t] = [r\ast t]$. We will construct analogous composition operators for lsos-comtraces and cd-graphs. We will then show that lsos-comtraces (cd-graphs) over a comtrace alphabet $\theta$ together with its composition operator form a monoid isomorphic to the comtrace monoid $(\E^*/\!\!\equiv_{\theta},\circledast,[\epsilon])$.

Given two sets $X_{1}$ and $X_{2}$, we  write $X_{1}\uplus X_{2}$ to denote the \emph{disjoint union} of $X_{1}$ and $X_{2}$. Such disjoint union can be easily obtained by renaming the elements in $X_{1}$ and $X_{2}$ so that $X_{1}\cap X_{2} =\emptyset$. We define the lsos-comtrace composition operator as follows.

\begin{definition}[composition of lsos-comtraces] Let $T_{1}$ and $T_{2}$ be lsos-comtraces over an alphabet $\theta=(E,sim,ser)$, where $T_{i}=\seq{X_{i},\prec_{i},\sqsubset_{i},\lambda_{i}}$. The \emph{composition} $T_{1}\odot T_{2}$ of $T_{1}$ and $T_{2}$ is defined as (a lp-isomorphic class of) a labeled so-structure $\seq{X,\prec,\sqsubset,\lambda}$ such that $X=X_{1}\uplus X_{2}$, $\lambda = \lambda_{1}\cup \lambda_{2}$, and $(X,\prec,\sqsubset) = \left(X,\PO_{\langle1,2\rangle},\sq_{\langle1,2\rangle}\right)^{\lozenge}$, where\smallskip\\
\mbox{\hspace{1.5cm}}$\PO_{\langle1,2\rangle}\;=\; \PO_{1}\cup \PO_{2}\cup\; \bset{(\alpha,\beta)\in X_{1}\times X_{2}\mid (\lambda(\alpha),\lambda(\beta))\not\in ser}$\\
\mbox{\hspace{1.5cm}}$\sq_{\langle1,2\rangle}\; =\; \sq_{1}\cup \sq_{2}\cup\; \bset{(\alpha,\beta)\in X_{1}\times X_{2}\mid (\lambda(\beta),\lambda(\alpha))\not\in ser}$ \EOD
\end{definition}

Observe that the operator is well-defined since  we can easily check that:
\begin{proposition}
For every $T_1,T_2\in \LCT(\theta)$, $T_1\odot T_2 \in \LCT(\theta)$. \qed 
\end{proposition}

We will next show that this composition operator $\_\odot\_$ properly corresponds to the operator $\_\circledast\_$  of the comtrace monoid over $\theta$.

\begin{proposition}  Let $\theta$ be a comtrace alphabet. Then  
\begin{enumerate}
\item For every $R,T \in \LCT(\theta)$, $\lc(R \odot T) = \lc(R) \circledast \lc(T).$

\item For every $\mathbf{r},\mathbf{t} \in \quotient{\E_\theta^*}{\equiv_{\theta}}$,
$\cl(\mathbf{r}\circledast \mathbf{t})= \cl(\mathbf{r})\odot\cl(\mathbf{t}).$
\end{enumerate}
\label{prop:hom1}
\end{proposition}
\begin{proof}\textbf{1. }Assume $R=\seq{X_1,\PO_1,\sq_1,\lambda_1}$, $T=\seq{X_2,\PO_2,\sq_2,\lambda_2}$ and $Q=\seq{X_1\uplus X_1,\PO,\sq,\lambda}$. We can pick $\lhd_1\in ext(R)$ and $\lhd_2\in ext(T)$. Then observe that a stratified order $\lhd$ satisfying $\Omega_{\lhd} = \Omega_{\lhd_1}\ast \Omega_{\lhd_2}$ is an extension of $Q$. Thus, by \theoref{rep}, we have 
$\lc(R) \circledast \lc(T) = [\map(\lambda_1,\lhd_1)]\circledast [\map(\lambda_2,\lhd_2)] =  [\map(\lambda,\lhd)] = \lc(Q)$ as desired.

\textbf{2. } Without loss of generality, we can assume that $\mathbf{r}= [r]$, $\mathbf{t}=[t]$ and $\mathbf{q}=[q] = \mathbf{r}\circledast \mathbf{t}$, where $q = r \ast t$. By appropriate reindexing, we can also assume that $\Sigma_{\mathbf{q}} = \Sigma_{\mathbf{r}}\uplus \Sigma_{\mathbf{t}}$. Under these assumptions, let $\cl(\mathbf{r})=T_{1} = \seq{\Sigma_{\mathbf{r}},\PO_{\mathbf{r}},\sq_{\mathbf{r}},l_{1}}$, $\cl(\mathbf{t})=T_{2} = \seq{\Sigma_{\mathbf{t}},\PO_{\mathbf{t}},\sq_{\mathbf{t}},l_{2}}$ and $\cl(\mathbf{q}) = T = \seq{\Sigma_{\mathbf{q}},\PO_{\mathbf{q}},\sq_{\mathbf{q}},l}$, where $l=l_{1}\cup l_{2}$ is simply the standard labeling functions. It will now suffice to show that $T_{1}\odot T_{2} = T$.

($\subseteq$): Let $T_{1}\odot T_{2} = (\Sigma_{\mathbf{r}}\uplus \Sigma_{\mathbf{t}},\PO_{\langle \mathbf{r},\mathbf{t}\rangle},\sq_{\langle \mathbf{r},\mathbf{t}\rangle},l)^{\lozenge}$. By Definitions \ref{def:s2inv} and \ref{def:s2sos}, we have\\
\mbox{\hspace{1.5cm}}$\PO_{\langle \mathbf{r},\mathbf{t}\rangle}\;=\; \PO_{\mathbf{r}}\cup \PO_{\mathbf{t}}\cup\; \bset{(\alpha,\beta)\in \Sigma_{\mathbf{r}}\times \Sigma_{\mathbf{t}}\mid (\lambda(\alpha),\lambda(\beta))\not\in ser} \subseteq\, \PO_{\mathbf{q}}$\\
\mbox{\hspace{1.5cm}}$\sq_{\langle \mathbf{r},\mathbf{t}\rangle}\;=\; \sq_{\mathbf{r}}\cup \sq_{\mathbf{t}}\cup\; \bset{(\alpha,\beta)\in \Sigma_{\mathbf{r}}\times \Sigma_{\mathbf{t}}\mid (\lambda(\beta),\lambda(\alpha))\not\in ser}\subseteq\, \sq_{\mathbf{q}}$\\
Thus, by \propref{so-cl} (5),  we have $(\Sigma_{\mathbf{r}}\uplus \Sigma_{\mathbf{t}},\PO_{\langle \mathbf{r},\mathbf{t}\rangle},\sq_{\langle \mathbf{r},\mathbf{t}\rangle},l)^{\lozenge}\subseteq  (\Sigma_{\mathbf{q}},\PO_{\mathbf{q}},\sq_{\mathbf{q}},l)$ as desired. Furthermore, by  \propref{so-cl} (5),  $(\Sigma_{\mathbf{r}}\uplus \Sigma_{\mathbf{t}},\PO_{\langle \mathbf{r},\mathbf{t}\rangle},\sq_{\langle \mathbf{r},\mathbf{t}\rangle})^{\lozenge}$ is a so-structure.

($\supseteq$): By Definitions \ref{def:s2inv} and \ref{def:s2sos}, we have $\PO_{q}\subseteq\, \PO_{\langle \mathbf{r},\mathbf{t}\rangle}$ and $\sq_{q}\subseteq\, \sq_{\langle \mathbf{r},\mathbf{t}\rangle}$. Since we already know $(\Sigma_{\mathbf{r}}\uplus \Sigma_{\mathbf{t}},\PO_{\langle \mathbf{r},\mathbf{t}\rangle},\sq_{\langle \mathbf{r},\mathbf{t}\rangle})^{\lozenge}$ is a so-structure, it follows from \propref{so-cl} (5) that \\
\mbox{\hspace{1.3cm}}$(\Sigma_{\mathbf{q}},\PO_{\mathbf{q}},\sq_{\mathbf{q}},l) = (\Sigma_{\mathbf{q}},\PO_{q},\sq_{q},l)^{\lozenge}\subseteq (\Sigma_{\mathbf{r}}\uplus \Sigma_{\mathbf{t}},\PO_{\langle \mathbf{r},\mathbf{t}\rangle},\sq_{\langle \mathbf{r},\mathbf{t}\rangle},l)^{\lozenge} = T_{1}\odot T_{2}.$\qed
\end{proof}

Let $\mathbb{I}$ denote the lp-isomorphic class  $[\emptyset,\emptyset,\emptyset,\emptyset]$. Then we observe that $\cl([\epsilon]) = \mathbb{I}$ and $\lc(\mathbb{I}) = [\epsilon]$. By \propref{hom1} and \theoref{rep}, the structure $(\LCT(\theta),\odot, \mathbb{I})$ is isomorphic to the monoid $(\quotient{\E_\theta^*}{\equiv_{\theta}},\circledast,[\epsilon])$ under the isomorphisms $\cl:\quotient{\E_\theta^*}{\equiv_{\theta}} \rightarrow \LCT(\theta)$ and $\lc:\LCT(\theta)\rightarrow \quotient{\E_\theta^*}{\equiv_{\theta}}$. Thus, the triple $(\LCT(\theta),\odot , \mathbb{I})$ is also a monoid. We can summarize these facts in the following theorem:
\begin{theorem}
The mappings $\cl$ and $\lc$ are monoid isomorphisms between two monoids $(\quotient{\E_\theta^*}{\equiv_{\theta}},\circledast,[\epsilon])$ and  $(\LCT(\theta),\odot, \mathbb{I})$. \qed
\end{theorem}

Similarly, we can also define a composition operator for cd-graphs.
\begin{definition}[composition of cd-graphs] Let $D_{1}$ and $D_{2}$ be cd-graphs over an alphabet $\theta=(E,sim,ser)$, where $D_{i}=\seq{X_{i},\cau_{i},\wcau_{i},\lambda_{i}}$. The \emph{composition} $D_{1}\circledcirc D_{2}$ of $D_{1}$ and $D_{2}$ is defined as (a lp-isomorphic class of) a labeled so-structure $\seq{X,\cau,\wcau,\lambda}$ such that $X=X_{1}\uplus X_{2}$, $\lambda = \lambda_{1}\cup \lambda_{2}$, and \smallskip\\
\mbox{\hspace{1.5cm}}$\cau\;=\; \cau_{1}\cup \cau_{2}\cup\; \set{(\alpha,\beta)\in X_{1}\times X_{2}\mid (\lambda(\alpha),\lambda(\beta))\not\in ser}$\\
\mbox{\hspace{1.5cm}}$\wcau\;\, =\,\; \wcau_{1}\cup \wcau_{2}\cup\; \set{(\alpha,\beta)\in X_{1}\times X_{2}\mid (\lambda(\beta),\lambda(\alpha))\not\in ser}$ \EOD
\end{definition}

From this definition, it is straightforward to show the following propositions, which we will state without proofs.

\begin{proposition}
For every $D_1,D_2\in \DCD(\theta)$, $D_1\circledcirc D_2 \in \DCD(\theta)$. \qed 
\end{proposition}

\begin{proposition}  Let $\theta$ be a comtrace alphabet. Then  
\begin{enumerate}
\item For every $R,T \in \LCT(\theta)$, $\ld(R \odot T) = \ld(R) \circledcirc \ld(T).$
\item For every $D,E \in \DCD(\theta)$, $\dl(D \circledcirc E) = \dl(D) \odot \dl(E).$ \qed
\end{enumerate}
\label{prop:hom2}
\end{proposition}

Putting the two preceding propositions  and \theoref{deprep} together, we conclude:
\begin{theorem}
The mappings $\ld$ and $\dl$ are monoid isomorphisms between two monoids $(\LCT(\theta),\odot, \mathbb{I})$ and $(\DCD(\theta),\circledcirc, \mathbb{I})$. \qed
\end{theorem}

\section{Conclusion}

The simple yet useful construction we used extensively in this paper is to build a quotient so-structure modulo the $\sq$-cycle equivalence relation. Intuitively, each $\sq$-cycle equivalence class consists of all the events that must be executed simultaneously with one another and hence can be seen as a single ``composite event''.  The resulting quotient so-structure  is technically  easier to handle since both relations of the quotient so-structure are acyclic. From this construction, we were able to give a labeled so-structure definition for comtraces similar to the labeled poset definition for traces. This quotient construction also explicitly  reveals the following connection: a step on a step sequence $s$ is not serializable with respect to  the relation $ser$ of a comtrace alphabet if and only if it corresponds to a $\sq$-cycle equivalence class of the lsos-comtrace representing the comtrace $[s]$ (cf. \propref{covlsos}). 

We have also formally shown that the quotient monoid of comtraces, the monoid of lsos-comtraces and the monoid of cd-graphs \emph{over the same comtrace alphabet} are indeed isomorphic by establishing monoid isomorphisms between them. These three models are  formal linguistic,  order-theoretic, and  graph-theoretic respectively, which allows us to apply a variety of tools and techniques. 

An immediate future task is to develop a framework similar to the one in this paper for \emph{generalized comtraces}, proposed and developed in \cite{JL08,Le,JL09}. Generalized comtraces extend comtraces with the ability to model events that can be executed \emph{earlier than or later than but never simultaneously}. Another direction is to define and analyze infinite comtraces (and generalized comtraces) in a spirit similar  to the works on infinite traces, e.g., \cite{Gas90,Die91}. It is also promising to use infinite lsos-comtraces and cd-graphs  to develop logics for comtraces similarly to what have been done for traces (cf. \cite{TW02,DHK07}). 

\subsubsection*{Acknowledgments.} I am grateful to Prof. Ryszard Janicki for introducing me  comtrace theory. I also thank the Mathematics Institute of Warsaw University and the Theoretical Computer Science Group of Jagiellonian University for their supports during my visits. It was during these visits that the ideas from this paper emerge. This work is financially  supported by the Ontario Graduate Scholarship and the Natural Sciences and Engineering Research Council of Canada. The anonymous referees are thanked for their valuable comments that help improving the readability of this paper.

\appendix
\section{Proof of \propref{covlsos}}
\begin{proof}\textbf{1.} ($\Rightarrow$):  Since $[\alpha]=[\beta]$,  we know that $\alpha=\beta$ or $(\alpha\sq \beta\wedge \beta \sq \alpha)$. The former case is trivial. For the latter case, by \theoref{SzpStrat},  we have
$\forall \lhd \in ext(S).\; \alpha \lhd^{\frown} \beta$ and $\forall \lhd \in ext(S).\; \beta \lhd^{\frown} \alpha$. 
But this implies that $\forall \lhd \in ext(S).\; \alpha \frown_{\lhd} \beta$.

($\Leftarrow$): The case when $\alpha=\beta$ is trivial. Assume that  $\alpha\not=\beta$ and $\forall \lhd \in ext(S).\; \alpha \frown_{\lhd} \beta$. Thus, by \theoref{SzpStrat}, $\alpha \sq \beta$ and $\beta \sq \alpha$. But this means $\alpha$ and $\beta$ belong to the same equivalence class.

\textbf{2. } Suppose for a contradiction that all $\lhd \in ext(S)$ cannot be written in the form of $\Omega_{\lhd} = u[\alpha]v$. This implies that there exists some $\gamma\in X\setminus [\alpha]$ such that for all $\lhd \in ext(S)$, $\alpha \frown_\lhd \gamma$. But by \theoref{SzpStrat}, this yields $\alpha \sq \gamma$ and $\gamma \sq \alpha$, contradicting with $\gamma\not\in [\alpha]$.

\textbf{3. } Assume $[\alpha] \imm{\hat{\sq}} [\beta]$. Suppose for a contradiction that there does not exist $\lhd \in ext(S)$ such that $\Omega_{\lhd} = u[\alpha][\beta]v$ for some step sequences $u$ and $v$. Then, by  \theoref{SzpStrat},  there must exist some $\gamma \in X\setminus ([\alpha]\cup[\beta])$, such that $\alpha \sq \gamma \sq \beta$. Since $\gamma \not\in [\alpha]\cup[\beta]$, this yields $[\alpha]\sq [\gamma] \sq [\beta]$, which  contradicts that $[\alpha] \imm{\hat{\sq}} [\beta]$.
\qed
\end{proof}

\section{Proof of \propref{validss}}
\begin{proof}\textbf{1. } Assume $\alpha,\beta \in A_i$ and $\alpha\not=\beta$. Thus, $\alpha\frown_{\lhd}\beta$.  Thus, by \corref{SzpStrat} (2), we have $\alpha \frown_\PO \beta$. Hence, by \textsf{LC5} of \defref{lcomtrace},  $(\lambda(\alpha),\lambda(\beta))\in sim$. Since $sim$ is irreflexive, this also shows that any two distinct $\alpha$ and $\beta$ in $A_i$ have different labels. Thus, $|A_{i}|=|\lambda(A_{i})|$ for all $i$.

\textbf{2. } From the proof of \textbf{1.}, we know that  $\alpha,\beta \in A_i$ and $\alpha\not=\beta$ implies $(\lambda(\alpha),\lambda(\beta))\in sim$. Thus, $\lambda(A_{i})\in \E_{\theta}$ for all $i$.
\qed
\end{proof}

\section{Proof of \propref{uniquexi}}
\begin{proof} Observe that from \propref{validss}, we have $\Sigma_{v}=\Sigma_{w}$. It remains to show that  $\xi_{v}(\alpha) = \xi_{w}(\alpha)$ for all $\alpha\in \Sigma_{v}$. 

Suppose for a contradiction that $\xi_{v}(\alpha) \not= \xi_{w}(\alpha)$ for some $\alpha\in \Sigma_{v}$. From the definition above, there are two distinct elements $x,y\in X$, such that $\xi_{v}(\alpha)= x$ and $\xi_{w}(\alpha)= y$ and $\lambda(x)= \lambda(y)=\ell(\alpha)$. Since $ser$ is irreflexive, $(\lambda(x),\lambda(y))\not \in ser \cup ser^{-1}$. Thus, by \textsf{LC4} of \defref{lcomtrace}, $x\PO y$ or $y \PO x$. Without loss of generality, we assume $x\PO y$ and that $\alpha = a^{(i)}$ for some event $a\in E$. 

Again by \textsf{LC4} of \defref{lcomtrace}, we know that elements  having the same label are totally ordered by $\PO$. Thus, if $k$ is the number of elements in $X$ labeled by $a$, then we have $\xi_{w}(a^{(1)})\PO \xi_{w}(a^{(2)})\PO \ldots  \PO \xi_{w}(a^{(k)})$ and  $\xi_{v}(a^{(1)})\PO \xi_{v}(a^{(2)})\PO \ldots  \PO \xi_{v}(a^{(k)})$. But then  $\xi_{v}(a^{(i)})= x$ implies that $|\set{z\in X\mid z\PO y \;\wedge\;  \lambda(z) = a}| \ge i$, while $\xi_{w}(a^{(i)})= y$ implies that $|\set{z\in X\mid z\PO y \;\wedge\;  \lambda(z) = a}| < i$, which is absurd. \qed
\end{proof}

\section{Proof of \lemref{l3}}
\begin{proof} Let $T=\bigl(\Sigma_{[u]},\PO_{[u]},\sq_{[u]},l \bigr)$. From \theoref{com2sos}, $T$ is a labeled so-structure. It only remains to show that $T$ satisfies conditions \textsf{LC1}--\textsf{LC5} of \defref{lcomtrace}. 

\textsf{LC1}: Assume  $[\alpha] (\imm{\hat{\sq}}\cap \hat{\PO}) [\beta]$ and suppose for a contradiction that $\lambda([\alpha])\times \lambda([\beta]) \subseteq ser$. Then from \propref{covlsos} (3), there exists $\lhd \in ext(T)$ such that $\Omega_{\lhd} = v[\alpha][\beta]w$. From \theoref{com2sos}, since we have $\lhd \in  \bset{\lhd_{s}\mid s \in [u]} = ext(S_{[u]})$, it follows that $\map(l,v[\alpha][\beta]w) \in [u]$. But  $\lambda([\alpha])\times \lambda([\beta]) \subseteq ser$ implies $\map(\ell,u[\alpha]\cup [\beta]v) \in [u]$. Hence,  $u[\alpha]\cup [\beta]v$ is also a stratified extension of $T$, which contradicts that $[\alpha] \hat{\PO} [\beta]$. Using a similar argument, we can show \textsf{LC2} using \propref{covlsos} (1,4) and \textsf{LC3} using \propref{covlsos} (1,2).  

\textsf{LC4}: Follows from Definitions \ref{def:s2inv} and \ref{def:s2sos} and the $\lozenge$-closure definition.

\textsf{LC5}: Since $\alpha \frown_{\PO}\beta$,  it follows from \corref{SzpStrat} that there exists $\lhd \in ext(T)$ where $\alpha \frown_\lhd \beta$. Since $\bset{\lhd_{s}\mid s \in [u]} = ext(S_{[u]})$, there exists a sequence $s\in [u]$ such that $s = \map(f,\Omega_\lhd)$. This implies $\alpha$ and $\beta$ belong to the same step in  $\h{s}$. Thus, we have $(\lambda(\alpha),\lambda(\beta))\in sim$. \qed
\end{proof}

\end{document}